\documentclass[aps, reprint, onecolumn, nofootinbib, longbibliography]{revtex4-2}
\usepackage{amsmath, amssymb, amsthm, MnSymbol}
\usepackage[margin=2cm]{geometry}
\usepackage[utf8]{inputenc}
\usepackage{booktabs}       % professional-quality tables
\usepackage{amsfonts}       % blackboard math symbols
\usepackage{nicefrac}       % compact symbols for 1/2, etc.
\usepackage{microtype}      % microtypography
\usepackage{graphicx}
\usepackage{bm,bbm}% bold math
\usepackage{xcolor}
\usepackage{todonotes}
\usepackage{hyperref}
\usepackage{lipsum}

\newtheorem{theorem}{Theorem}[]
\newtheorem{corollary}{Corollary}[theorem]

\newtheorem{definition}{Definition}[]

\newcommand{\bpm}{\begin{pmatrix}}
\newcommand{\epm}{\end{pmatrix}}
\newcommand{\ket}[1]{\ensuremath{\left| #1 \right \rangle}}
\newcommand{\bra}[1]{\ensuremath{\left \langle #1 \right |}}
\newcommand{\braket}[2]{\ensuremath{\left\langle #1\left|#2 \right.\right\rangle}}
\newcommand{\ketbra}[2]{\ket{#1}\!\bra{#2}}

\newcommand{\tr}[1]{ \mathrm{tr}\left[#1\right] }

\newcommand{\kett}[1]{\ensuremath{\left | #1  \right \rrangle }}
\newcommand{\braa}[1]{\ensuremath{\left \llangle  #1 \right |}}
\newcommand{\braakett}[2]{\ensuremath{\left\llangle   #1\left|#2 \right.\right \rrangle}}
\newcommand{\kettbraa}[2]{\kett{#1}\! \braa{#2}}
\newcommand{\meas}{\mathcal{M}}

\newtheorem{example}{Example}[section]

\begin{document}
\title{Supervised quantum machine learning models are kernel methods}
\author{Maria Schuld}
\affiliation{Xanadu, Toronto, ON, M5G 2C8, Canada}

\begin{abstract}
    With near-term quantum devices available and the race for fault-tolerant quantum computers in full swing, researchers became interested in the question of what happens if we replace a supervised machine learning model with a quantum circuit. While such ``quantum models'' are sometimes called ``quantum neural networks'', it has been repeatedly noted that their mathematical structure is actually much more closely related to kernel methods: they analyse data in high-dimensional Hilbert spaces to which we only have access through inner products revealed by measurements. This technical manuscript summarises and extends the idea of systematically rephrasing supervised quantum models as a kernel method. With this, a lot of near-term and fault-tolerant quantum models can be replaced by a general support vector machine whose kernel computes distances between data-encoding quantum states. Kernel-based training is then guaranteed to find better or equally good quantum models than variational circuit training. Overall, the kernel perspective of quantum machine learning tells us that the way that data is encoded into quantum states is the main ingredient that can potentially set quantum models apart from classical machine learning models. 
\end{abstract}

\maketitle

\section{Motivation}

\begin{figure}[b]
    \centering
    \includegraphics[width=0.9\textwidth]{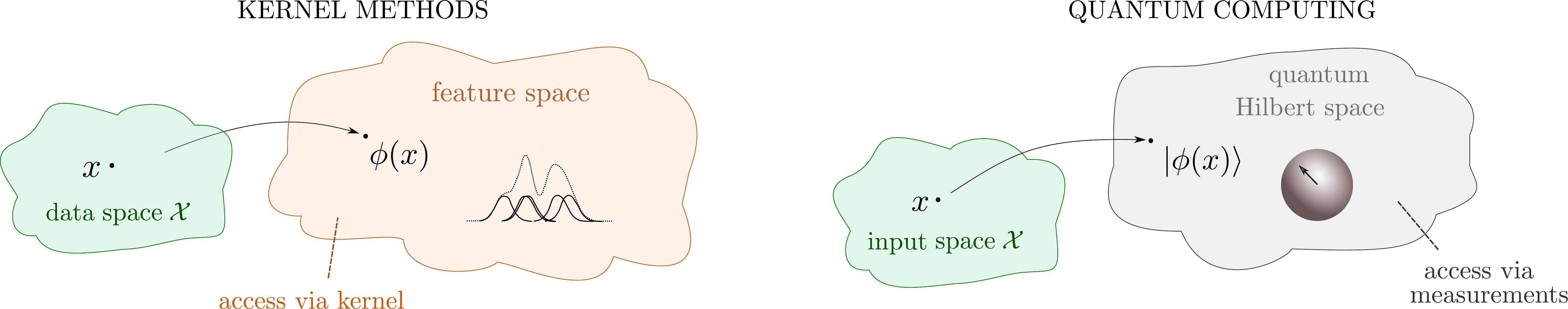}
    \caption{\textbf{Quantum computing and kernel methods are based on a similar principle.} Both have mathematical frameworks in which information is mapped into and then processed in high-dimensional spaces to which we have only limited access. In kernel methods, the access to the feature space is facilitated through \textit{kernels} or inner products of feature vectors. In quantum computing, access to the Hilbert space of quantum states is given by measurements, which can also be expressed by inner products of quantum states.}
    \label{fig:connection}
\end{figure}

The mathematical frameworks of quantum computing and kernel methods are strikingly similar: both describe how information is processed by mapping it to vectors that live in potentially inaccessibly large spaces, without the need of ever computing an explicit numerical representation of these vectors (Figure~\ref{fig:connection}). This similarity is particularly obvious -- and as we will see, useful -- in \textit{quantum machine learning}, an emerging research field that investigates how quantum computers can learn from data \citep{wittek2014quantum, biamonte2017quantum, schuld2018supervised}. If the data is ``classical'' as in standard machine learning problems, quantum machine learning algorithms have to encode it into the physical states of quantum systems. This process is formally equivalent to a \textit{feature map} that assigns data to quantum states (see \cite{schuld2019quantum, havlivcek2019supervised} but also earlier notions in \cite{rebentrost2014quantum, chatterjee2016generalized, schuld2018circuit}). Inner products of such data-encoding quantum states then give rise to a kernel, a kind of similarity measure that forms the core concept of kernel theory.  

The natural shape of this analogy sparked more research in the past years, for example on training generative quantum models \citep{liu2018differentiable}, constructing kernelised machine learning models \citep{blank2020quantum}, understanding the separation between the computational complexity of quantum and classical machine learning \citep{havlivcek2019supervised, liu2020rigorous, huang2020power} or revealing links between quantum machine learning and maximum mean embeddings \citep{kubler2019quantum} as well as metric learning \citep{lloyd2020quantum}. But despite the growing amount of literature, a comprehensive review of the link between quantum computation and kernel theory, as well as its theoretical consequences, is still lacking. This technical manuscript aims at filling the gap by summarising, formalising and extending insights scattered across existing literature and ``quantum community folklore''. The central statement of this line of work is that quantum algorithms optimised with data can \textit{fundamentally} be formulated as a classical kernel method whose kernel is computed by a quantum computer. This statement holds both for the popular class of classically trained variational near-term algorithms (e.g., \cite{benedetti2019parameterized}) as well as for more sophisticated fault-tolerant algorithms \textit{trained} by a quantum computer (e.g., \cite{rebentrost2014quantum}). It will be apparent that once the right ``spaces'' for the analysis are defined (as first proposed in \cite{havlivcek2019supervised}), the theory falls into place itself. This is in stark contrast to the more popular, but much less natural, attempt to force quantum theory into the shape of neural networks.\footnote{In some sense, many near-term approaches to quantum machine learning can be understood as a kernel method with a special kind of kernel, where the model (and possibly even the kernel \citep{lloyd2020quantum}) are trained like neural networks. This mix of both worlds makes quantum machine learning an interesting mathematical playground beyond the questions of asymptotic speedups that quantum computing researchers tend to ask by default. } 

A lot of the results presented here are of theoretical nature, but have important practical implications. Understanding quantum models as kernel methods means that the expressivity, optimisation and generalisation behaviour of quantum models is largely defined by the data-encoding strategy or \textit{quantum embedding} which fixes the kernel. Furthermore, it means that while the kernel itself may explore high-dimensional state spaces of the quantum system, quantum models can be trained and operated in a low-dimensional subspace. In contrast to the popular strategy of variational models (where a quantum algorithm depends on a tractable number of classical parameters that are optimised), we do not have to worry about finding the right variational circuit ansatz, or about how to avoid \textit{barren plateaus} problems \cite{mcclean2018barren, holmes2021connecting} -- but pay the price of having to compute pairwise distances between data points. 

For classical machine learning research, the kernel perspective can help to demystify quantum machine learning. A medium-term benefit may also derive from quantum computing's extensive tools that describe information in high-dimensional spaces, and possibly from interesting new kinds of kernels derived from physics. In the longer term, quantum computers promise access to fast linear algebra processing capabilities which are in principle able to deliver the polynomial speed-up that allows kernel methods to process big data without relying on approximations and heuristics. 

The manuscript is aimed at readers coming from \textit{either} a machine learning \textit{or} quantum computing background, but assumes an advanced level of mathematical knowledge of Hilbert spaces and the like (and there will be a lot of Hilbert spaces). Instead of giving lengthy introductions to both fields at the beginning, I will try to explain relevant concepts such as quantum states, measurements, or kernels as they are needed. Since neither kernel methods nor quantum theory are easy to digest, the next section will summarise all the main insights from a high-level point of view to connect the dots right from the start. 

A final side note may be useful: quantum computing researchers love to precede any concept with the word ``quantum''. In a young and explorative discipline like quantum machine learning (there we go!), this leads to very different ideas being labeled as ``quantum kernels'', ``quantum support vector machines'', ``quantum classifiers'' or even ``quantum neural networks''. To not add to this state of confusion I will -- besides standard technical terms -- only use the ``quantum'' prefix if a quantity is explicitly computed by a quantum algorithm (instead of being a mathematical construction in quantum theory). I will therefore speak of ``quantum models'' and ``quantum kernels'', but try to avoid constructions like ``quantum feature maps'' and ``quantum reproducing kernel Hilbert space''. 

\section{Summary of results}

\begin{figure*}[t]
    \centering
    \includegraphics[width=0.7\textwidth]{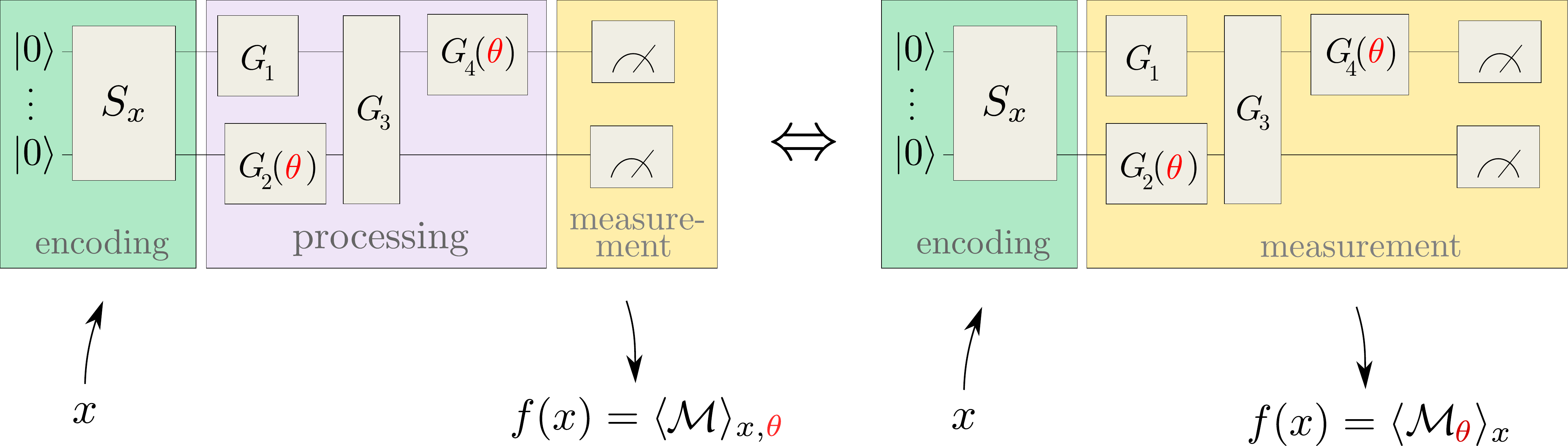}
    \caption{\textbf{Interpreting a quantum circuit as a machine learning model.} After encoding the data with the routine $S_x$, a quantum circuit ``processes'' the embedded input, followed by a measurement (left). The processing circuit may depend on classically trainable parameters, as investigated in near-term quantum machine learning with \textit{variational circuits}, or it may consist of standard quantum routines such as amplitude amplification or quantum Fourier transforms. The expected outcome of the measurement $\mathcal{M}$ is interpreted as the model's prediction, which is deterministic (generative models, which would consider the measurement \textit{samples} as outputs, are not considered here). Since the processing circuit only changes the basis in which the measurement is taken, it can conceptually be understood as part of the measurement procedure (right). In this sense, quantum models consist of two parts, the data encoding/embedding and the measurement. Training a quantum model is the problem of finding the measurement that minimises a data-dependent cost function. Note that while the measurement could depend on trainable parameters I will not consider trainable embedding circuits here.}
    \label{fig:variational}
\end{figure*}

First, a quick overview of the scope. Quantum algorithms have been proposed for many jobs in supervised machine learning, but the majority of them replace the \textit{model}, such as a classifier or generator, with an algorithm that runs on a quantum computer. These algorithms -- I will call them \textit{quantum models} -- usually consist of two parts: the data encoding, which maps data inputs $x$ to quantum states $\ket{\phi(x)}$ (effectively \textit{embedding} them into the space of quantum states), and a measurement $\meas$. Statistical properties of the measurement are then interpreted as the output of the model. Training a quantum model means to find the measurement which minimises a cost function that depends on training data. This overall definition is fairly general, and it includes most near-term  supervised quantum machine learning algorithms as well as many more complex, fault-tolerant quantum algorithms (see Figure~\ref{fig:variational}). Throughout this manuscript I will interpret the expected measurement -- or in practice, the average over measurement outcomes -- as a prediction, but the results may carry over to other settings, such as generative quantum models (e.g., \citep{benedetti2019generative}). I will also consider the embedding fixed and not trainable as proposed in \cite{perez2020data, lloyd2020quantum}.

The bridge between quantum machine learning and kernel methods is formed by the observation that quantum models map data into a high-dimensional feature space, in which the measurement defines a linear decision boundary as shown in Figure~\ref{fig:linear_model}. Note that for this to hold we need to define the data-encoding density matrices $\rho(x) = \ketbra{\phi(x)}{\phi(x)}$ as the feature ``vectors''\footnote{The term \textit{feature vectors} derives from the fact that they are elements of a vector space, not that they are vectors in the sense of the space $\mathbb{C}^N$ or $\mathbb{R}^N$.} instead of the Dirac vectors $\ket{\phi(x)}$ (see Section~\ref{ssec:linear}). This was first proposed in Ref. \cite{havlivcek2019supervised}. Density matrices are alternative descriptions of quantum states as Hermitian operators which are handy because they can also express probability distributions over quantum states (in which case they are describing so-called \textit{mixed} instead of \textit{pure} states). We can therefore consider the space of complex matrices enriched with the Hilbert-Schmidt inner product as the feature space of a quantum model and state:
\begin{quote}
      \textit{1. Quantum models are linear models in the ``feature vectors'' $\rho(x)$.}
\end{quote}
As famously known from support vector machines \cite{steinwart2008support}, linear models in feature spaces can be efficiently evaluated and trained if we have access to inner products of feature vectors, which is a function $\kappa$ in two data points $x, x'$ called the \textit{kernel}. Kernel theory essentially uses linear algebra and functional analysis to derive statements about the expressivity, trainability and generalisation power of linear models in feature spaces directly from the kernel. For us this means that we can learn a lot about the properties of quantum models if we study inner products $\kappa(x, x') = \tr{ \rho(x') \rho(x) }$, or, for pure states, $\kappa(x, x') = |\braket{\phi(x')}{\phi(x)}|^2$ (see in particular Ref. \cite{huang2020power}). I will call these functions ``quantum kernels'.

\begin{figure}[t]
    \centering
    \includegraphics[width=0.4\textwidth]{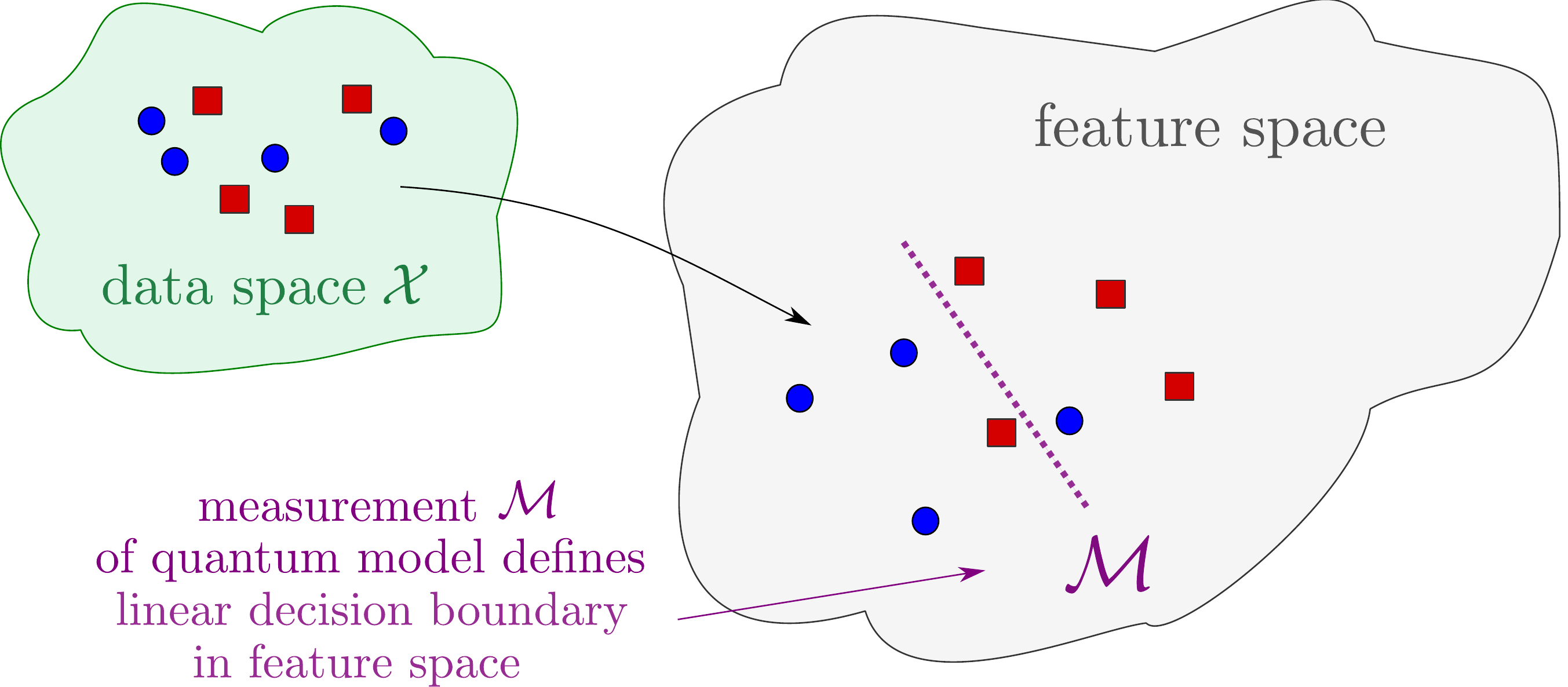}
    \caption{\textbf{Quantum models as linear models in a feature space.} A quantum model can be understood as a model that maps data into a feature space in which the measurement defines a linear decision boundary. This feature space is not identical to the Hilbert space of the quantum system. Instead we can define it as the space of complex matrices enriched with the Hilbert-Schmidt inner product -- which is the space where density matrices live in.}
    \label{fig:linear_model}
\end{figure}

To understand what kernels can tell us about quantum machine learning, we need another important concept from kernel theory: the \textit{reproducing kernel Hilbert space} (RKHS). An RKHS is an alternative feature space of a kernel -- and therefore reproduces all ``observable'' behaviour of the machine learning model. More precisely, it is a feature space of \textit{functions} $x \to g_x(\cdot) = \kappa(x, \cdot)$, which are constructed from the kernel. The RKHS contains one such function for every input $x$, as well as their linear combinations (for example, for the popular Gaussian kernel these linear combinations are sums of Gaussians centered in the individual data points). In an interesting -- and by no means trivial -- twist, these functions happen to be identical to the linear models in feature space. For quantum machine learning this means that the space of quantum models and the RKHS of the quantum kernel contain exactly the same functions (see Section~\ref{ssec:rkhs_quantum_models}). What we gain is an alternative representation of quantum models, one that only depends on the quantity $\tr{\rho(x') \rho(x)}$ (see Figure~\ref{fig:quantum_kernel_theory}).

\begin{figure*}[t]
    \centering
    \includegraphics[width=0.7\textwidth]{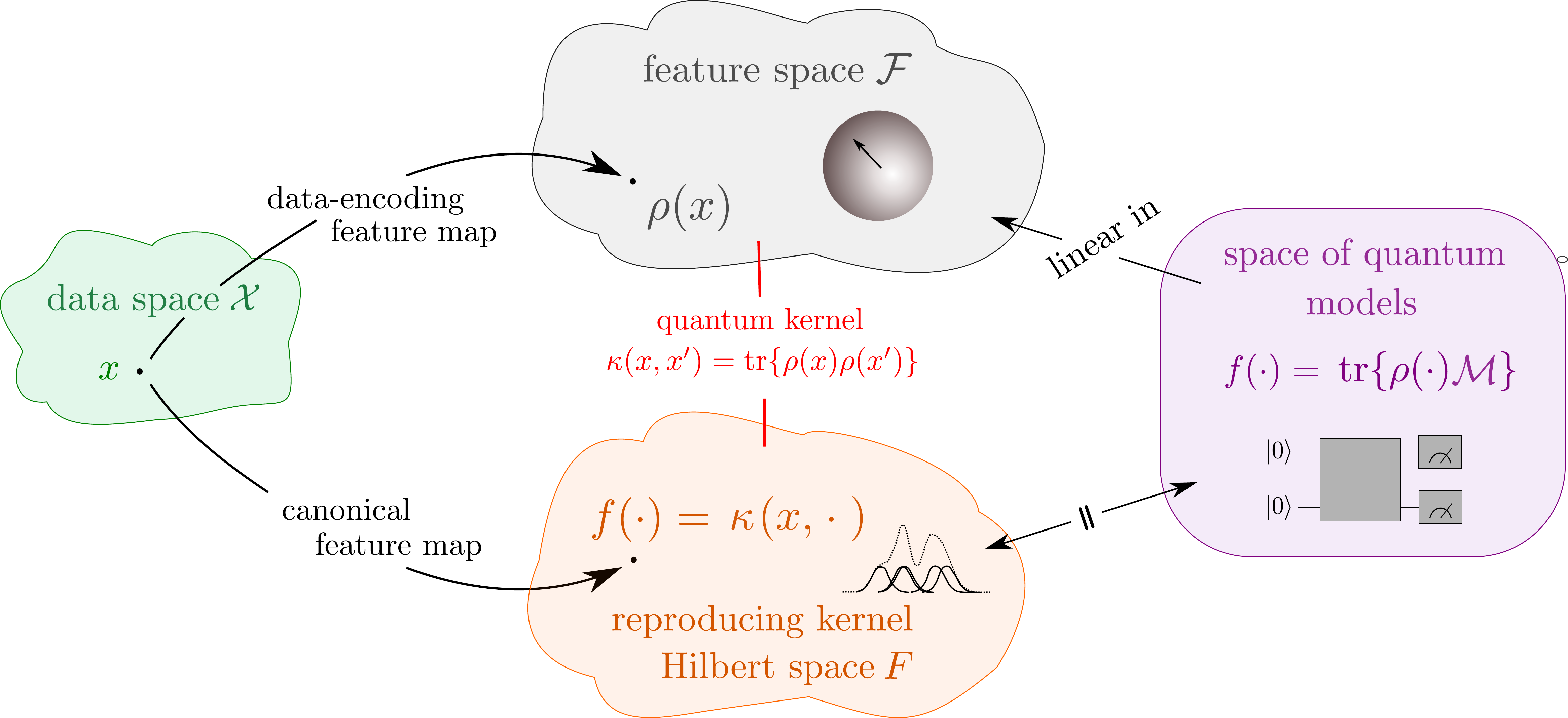}
    \caption{\textbf{Overview of the link between quantum models and kernel methods.} The strategy with which data is encoded into quantum states is a feature map from the space of data to the feature space $\mathcal{F}$ ``of density matrices'' $\rho$. In this space, quantum models can be expressed as a linear model whose decision boundary is defined by the measurement. According to kernel theory, an alternative feature space with the same kernel is the RKHS $F$, whose vectors are functions arising from fixing one entry of the kernel (i.e., the inner product of data-encoding density matrices). The RKHS is equivalent to the space of quantum models, which are linear models in the data-encoding feature space. These connections can be used to study the properties of quantum models as learners, which turn out to be largely determined by the kernel, and therefore by the data-encoding strategy. }
    \label{fig:quantum_kernel_theory}
\end{figure*}

This alternative representation can be very useful for all sorts of things. For example, it allows us to study the universality of quantum models as function approximators by investigating the universality of the RKHS, which in turn is a property of the quantum kernel. But probably the most important use is to study  optimisation: minimising typical cost functions over the space of quantum models is equivalent to minimising the same cost over the RKHS of the quantum kernel (see Section~\ref{ssec:optimsing_rkhs}). The famous \textit{representer theorem} uses this to show  that ``optimal models'' (i.e., those that minimise the cost) can be written in terms of the quantum kernel as 
\begin{equation}\label{eq:preview_optimal}
f_{\rm opt}(x) = \sum_{m=1}^M \alpha_m \tr{\rho(x^m) \rho(x)} =  \tr{ \left(\sum_{m=1}^M \alpha_m\rho(x^m) \right) \rho(x)},
\end{equation} 
where $x^m, m=1,\dots,M$ is the training data and $\alpha_m \in \mathbb{R}$ (see Section~\ref{ssec:meas_expansion}). Looking at the expression in the round brackets, this enables us to say something about optimal measurements for quantum models:
\begin{quote}
    \textit{2. Quantum models that minimise typical machine learning cost functions have measurements that can be written as ``kernel expansions in the data'', $\meas = \sum_m \alpha_m \rho(x^m)$.   }
\end{quote}
In other words, we are guaranteed that the best measurements for machine learning tasks only have $M$ degrees of freedom $\{\alpha_m\}$, rather than the $\mathcal{O}(2^{2n})$ degrees of freedom needed to express a general measurement on a standard $n$-qubit quantum computer. Even more, if we include a regularisation term into the cost function, the kernel defines entirely which models are actually penalised or preferred by regularisation. Since the kernel only depends on the way in which data is encoded into quantum states, one can conclude that data encoding fully defines the minima of a given cost function used to train quantum models (see Section~\ref{ssec:regularisation}).

But how can we \textit{find} the optimal model in Eq.~(\ref{eq:preview_optimal})? We could use the near-term approach to quantum machine learning and simply train an ansatz, hoping that it learns the right measurement. But as illustrated in Figure \ref{fig:training}, variational training typically only searches through a small subspace of all possible quantum models/measurements. This has a good reason: to train a circuit that can express any quantum model (and is hence guaranteed to find the optimal one) would require parameters for all $\mathcal{O}(2^{2n})$ degrees of freedom, which is intractable for all but toy models. However, also here kernel theory can help: not only is the optimal measurement defined by $M \ll 2^{2n}$ degrees of freedom, finding the optimal measurement has the same favourable scaling (see Section~\ref{ssec:convex}) if we switch to a kernel-based training approach.
\begin{quote} 
    \textit{3. The problem of finding the optimal measurement for typical machine learning cost functions trained with $M$ data samples can be formulated as an $M$-dimensional optimisation problem. }
\end{quote}
If the loss is convex, as is common in machine learning, the optimisation problem is guaranteed to be convex as well. Hence, under rather general assumptions, we are guaranteed that the ``hard'' problem of picking the best quantum model shown in Eq.~(\ref{eq:preview_optimal}) is tractable and of a simple structure, even without reverting to variational heuristics. In addition, convexity -- the property that there is only one global minimum -- may help with trainability problems like the notorious ``barren plateaus'' \citep{mcclean2018barren} in variational circuit training. If the loss function is the hinge loss, things reduce to a standard support vector machine with a quantum kernel, which is one of the algorithms proposed in \cite{schuld2019quantum} and \cite{havlivcek2019supervised}. 

\begin{figure*}[t]
    \centering
    \includegraphics[width=0.7\textwidth]{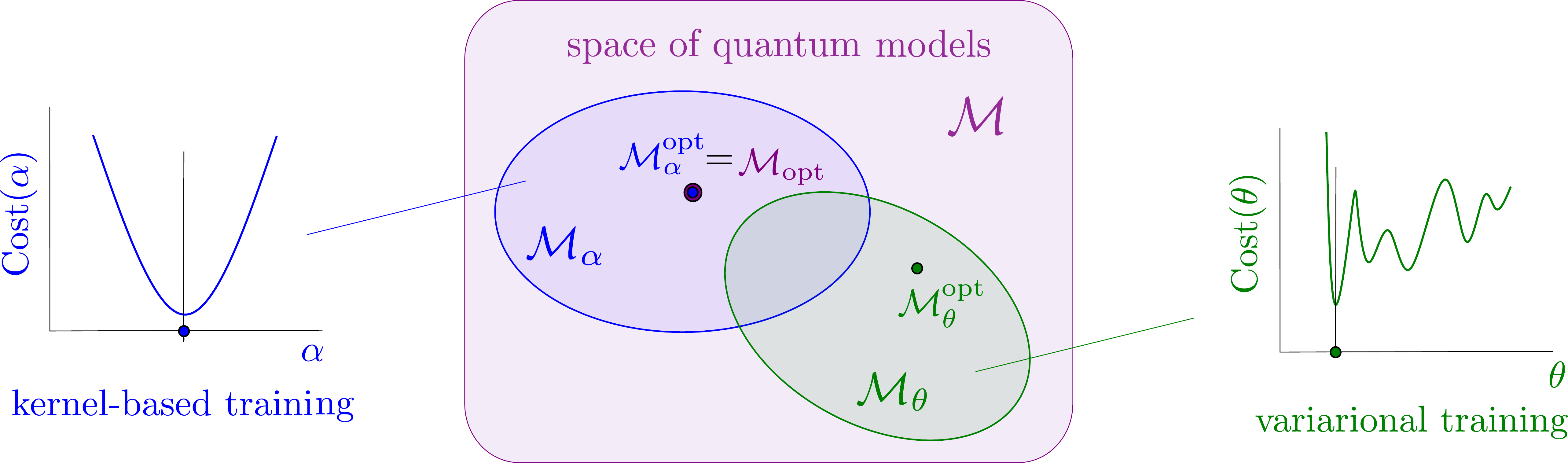}
    \caption{\textbf{Kernel-based training vs. variational training.} Training a quantum model as defined here tries to find the optimal measurement $\mathcal{M}_{\rm opt}$ over all possible quantum measurements. Kernel theory guarantees that in most cases this optimal measurement will have a representation that is a linear combination in the training data with coefficients $\alpha = (\alpha_1,\dots, \alpha_M)$. Kernel-based training therefore optimises over the parameters $\alpha$ directly, effectively searching for the best model in an $M$-dimensional subspace spanned by the training data (blue). We are guaranteed that $\mathcal{M}_{\alpha}^{\rm opt} = \mathcal{M}_{\rm opt}$, and if the loss is convex this is the only minimum, which means that kernel-based training will find the best measurement out of all measurements. Variational training parametrises the measurement instead by a general ansatz that depends on $K$ parameters $\theta = (\theta_1, \dots, \theta_K)$, and tries to find the optimal measurement $\mathcal{M}_{\theta}^{\rm opt}$ in the subspace explored by the ansatz. This $\theta$-subspace is not guaranteed to contain the globally optimal measurement  $\mathcal{M}_{\rm opt}$, and optimisation is usually non-convex. We are therefore guaranteed that kernel-based training finds better or the same minima to variational training, but at the expense of having to compute pairwise distances of data points for training and classification.
    }
    \label{fig:training}
\end{figure*}

Altogether, approaching quantum machine learning from a kernel perspective can have profound implications for the way we think about it.  Firstly, most quantum models can be formulated as general support vector machines (in the sense of \cite{steinwart2008support}) with a kernel evaluated on a quantum computer. As a corollary, we know that the measurements of optimal quantum models live in a low-dimensional subspace spanned by the training data, and that we can train in that space. Kernel-based training is guaranteed to find better minima -- or as phrased here, measurements -- than variational circuit training, at the expense of having to evaluate pair-wise distances of data points in feature space. (In the conclusion I will discuss how larger fault-tolerant quantum computers could potentially help with this as well!). Secondly, if the kernel defines the model, and the data encoding defines the kernel, we have to be very aware of the data encoding strategy we use in quantum machine learning -- a step that has often taken the backseat over other parts of quantum models. Thirdly, since quantum models can always be rewritten as a classical model plus quantum kernel, the separation between classical and quantum machine learning lies only in the ability of quantum computers to implement classically hard kernels. The first steps into investigating such separations have been made in papers like \cite{liu2020rigorous, huang2020power}, but it is still unclear whether any \textit{useful} applications turn out to be enabled solely by quantum computers.  

The remainder of the paper will essentially follow the structure of this synopsis to discuss every statement in more mathematical detail.

\section{Quantum computing, feature maps and kernels}\label{sec:definitions}

Let us start by laying the ground work for the kernel perspective on quantum machine learning. First I review the link between the process of encoding data into quantum states and feature maps, and construct the ``quantum kernel'' that we will use throughout the manuscript. I will then give some examples of data-encoding feature maps and quantum kernels, including a general description that allows us to understand these kernels via Fourier series. 

\subsection{Encoding data into quantum states is a feature map}\label{ssec:encoding_feat_map}

First, a few important concepts from quantum computing, which can be safely skipped by readers with a background in the field. Those who deem the explanations to be too casual shall be referred to the wonderful script by Michael Wolf \citep{wolf}.

\begin{quote}
    \textit{Quantum state.} According to quantum theory, the state of a quantum system is fully described by a length-1 vector $\ket{\psi}$ (or, more precisely, a ray represented by this vector) in a complex Hilbert space $\mathcal{H}$. The notation $\ket{\cdot}$ can be intimidating, but simply reminds of the fact that the Hilbert space has an inner product $\langle \cdot , \cdot \rangle$, which for Hilbert spaces describing quantum systems is denoted as $\braket{\cdot}{\cdot}$, and that its vectors constitute ``the right side'' of the inner product. Quantum theory textbooks then introduce the left side of the inner product as a functional $\bra{\varphi}$ from a dual space $\mathcal{H}^*$ acting on elements of the original Hilbert space. Mainstream quantum computing considers rather simple quantum systems of $n$ binary subsystems called ``qubits'', whose Hilbert space is the $ \mathbb{C}^{2^n}$. The dual space $\mathcal{H}^*$ can then be thought of as the space of complex $2^n$-dimensional ``row vectors''. A joint description of two quantum systems $\ket{\psi}$ and $\ket{\varphi}$ is expressed by the tensor product $\ket{\psi} \otimes \ket{\phi}$.\\
    
    \textit{Density matrix.} There is an alternative representation of a quantum state as a Hermitian operator called a \textit{density matrix}. The density matrix corresponding to a state vector $\ket{\psi}$ reads
    \begin{equation}
        \rho = \ketbra{\psi}{\psi}.
    \end{equation} 
    If we represent quantum states as vectors in $\mathbb{C}^{2^n}$, then the corresponding density matrix is given by the outer product of a vector with itself -- resulting in a matrix (and hence the name). The density matrix contains all observable information of $\ket{\psi}$, but is useful to model probability distributions $\{p_k\}$ over multiple quantum states $\{\ketbra{\psi_k}{\psi_k}\}$ as so-called \textit{mixed states}
    \begin{equation}
        \rho = \sum_k p_k \ketbra{\psi_k}{\psi_k},
    \end{equation} 
    without changing the equations of quantum theory. For simplicity I will assume that we are dealing with pure states in the following, but as far as I know everything should hold for mixed states as well.\\

    \textit{Quantum computations.} A quantum computation applies physical operations to quantum states, which -- in analogy to classical circuits -- are known as ``quantum gates''. The gates are applied to a small amount of qubits at a time. A collection of quantum gates (possibly followed by a measurement, which will be explained below) is called a \textit{quantum circuit}. Any physical operation acting on the quantum system maps from a density matrix $\rho$ to another density matrix $\rho'$. In the most basic setting, such a transformation is described by a unitary operator $U$, with $\rho' = U^{\dagger} \rho U$, or $\ket{\psi'} = U \ket{\psi}$.\footnote{The unitary operator is the quantum equivalent of a stochastic matrix which acts on vectors that represent discrete probability distributions.} Unitary operations are length-preserving linear transformations, which is why we often say that a unitary ``rotates'' the quantum state. In the finite-dimensional case, a unitary operator can conveniently be represented by a unitary matrix, and the evolution of a quantum state becomes a matrix multiplication.
\end{quote}

Consider a physical operation or quantum circuit $U(x)$ that depends on data $x \in \mathcal{X}$ from some data domain $\mathcal{X}$. For example, if the domain is the set of all bit strings of length $n$, the quantum circuit may apply specific operations only if bits are $1$ and do nothing if they are $0$. After the operation, the quantum state $\ket{\phi(x)} = U(x) \ket{\psi}$ depends on $x$. In other words, the data-dependent operation ``encodes'' or ``embeds'' $x$ into a vector $\ket{\phi(x)}$ from a Hilbert space (and I will use both terms interchangeably). This is a common definition of a feature map in machine learning, and we can say that \textit{any data-dependent quantum computation implements a feature map}. 

While from a quantum physics perspective it seems natural -- and has been done predominantly in the early literature -- to think of $x \rightarrow \ket{\phi(x)}$ as the feature map that links quantum computing to kernel methods, we will see below that quantum models are \textit{not} linear in the Hilbert space of the quantum system \cite{havlivcek2019supervised}, which means that the apparatus of kernel theory does not apply elegantly. Instead, I will define $x \to \rho(x)$ as the feature map and call it the \textit{data-encoding feature map}. Note that consistent with the proposed naming scheme, the term ``quantum feature map'' would be misleading, since the result of the feature map is a state, which without measurement is just a mathematical concept.

\begin{definition}[Data-encoding feature map]\label{def:data_enc_feat_map}
Given a $n$-qubit quantum system with states $\ket{\psi}$, and let $\mathcal{F}$ be the space of complex-valued $2^n \times 2^n$-dimensional matrices equipped with the Hilbert-Schmidt inner product $\langle \rho, \sigma \rangle_{\mathcal{F}} = \mathrm{tr} \{ \rho^{\dagger} \sigma \}$ for $\rho, \sigma \in \mathcal{F}$. The data-encoding feature map is defined as the transformation 
\begin{equation}
    \phi: \mathcal{X} \rightarrow \mathcal{F},
\end{equation}
\begin{equation} 
    \phi(x) = \ketbra{\phi(x)}{\phi(x)} = \rho(x),
\end{equation}
and can be implemented by a data-encoding quantum circuit $U(x)$.
\end{definition}
While density matrices of qubit systems live in a subspace of $\mathcal{F}$ (i.e., the space of positive semi-definite trace-class  operators), it will be useful to formally define the data-encoding feature space as above. Firstly, it makes sure that the feature space is a Hilbert space, and secondly, it allows measurements to live in the same space \cite{wolf}, which we will need to define linear models in $\mathcal{F}$. Section \ref{ssec:vectorised} will discuss that this definition of the feature space is equivalent to the tensor product space of complex vectors $\ket{\psi} \otimes \ket{\psi^*}$ used in \cite{huang2020power}.

\subsection{The data-encoding feature map gives rise to a kernel}

Let us turn to kernels.

\begin{quote}
    \textit{Kernels.} Unsurprisingly, the central concept of kernel theory are kernels, which in the context of machine learning are defined as real or complex-valued positive definite functions in two data points, $\kappa: \mathcal{X}\times \mathcal{X} \rightarrow \mathbb{K}$, where $\mathbb{K}$ can be $\mathbb{C} \text{ or }\mathbb{R}$. For every such function we are guaranteed that there exist at least one feature map such that inner products of feature vectors $\phi(x)$ from the feature Hilbert space $\mathcal{F}$ form the kernel, $\kappa(x, x') = \langle \phi(x'), \phi(x) \rangle_{\mathcal{F}}$. Vice versa, every feature map gives rise to a kernel. The importance of kernels for machine learning is that they are a means of ``computing'' in feature space without ever accessing or numerically processing the vectors $\phi(x)$: everything we need to do in machine learning can be expressed by inner products of feature vectors, instead of the feature vectors themselves. In the cases that are practically useful, these inner products can be computed by a comparably simple function. This makes the computations in intractably large spaces tractable.
\end{quote}

With the Hilbert-Schmidt inner product from Definition~\ref{def:data_enc_feat_map} we can immediately write down the kernel induced by the data-encoding feature map, which we will call the ``quantum kernel'' (since it is a function computed by a quantum computer):
\begin{definition}[Quantum kernel]\label{def:quantum_kernel}
Let $\phi$ be a data-encoding feature map over domain $\mathcal{X}$. A quantum kernel is the inner product between two data-encoding feature vectors $\rho(x), \rho(x')$ with $x, x' \in \mathcal{X}$,
\begin{equation}
    \kappa(x, x') = \tr{\rho(x') \rho(x)} = |\braket{\phi(x')}{\phi(x)}|^2.
\end{equation}
\end{definition}

To justify the term ``kernel'' we need to show that the quantum kernel is indeed a positive definite function. The quantum kernel is a product of the complex-valued kernel $\kappa_c(x, x') = \braket{\phi(x')}{\phi(x)}$ and its complex conjugate $\kappa_c(x, x')^* = \braket{\phi(x)}{\phi(x')} = \braket{\phi(x')}{\phi(x)}^*$. Since products of two kernels are known to be kernels themselves, we only have to show that the complex conjugate of a kernel is also a kernel. 
For any $x^m \in \mathcal{X}, m=1 \dots M$, and for any $c_m \in \mathbb{C}$, we have
\begin{align*}
    \sum_{m, m'} c_m c^*_{m'} \left(\kappa_c(x^m, x^{m'})\right)^* &= \sum_{m, m'} c_m c^*_{m'} \braket{\phi(x^{m})}{\phi(x^{m'})}\\
    &= \left(\sum_m c_m \bra{\phi(x^m)} \right) \left( \sum_m c^*_m \ket{\phi(x^m)}\right)\\
    &=  \Vert \sum_m c^*_m \ket{\phi(x^m)}  \Vert^2\\
    &\geq 0,
\end{align*}
which means that the complex conjugate of a kernel is also positive definite.

\begin{example}\label{ex:quantum_kernel}
\begin{figure}
    \centering
    \includegraphics[width = 0.3\textwidth]{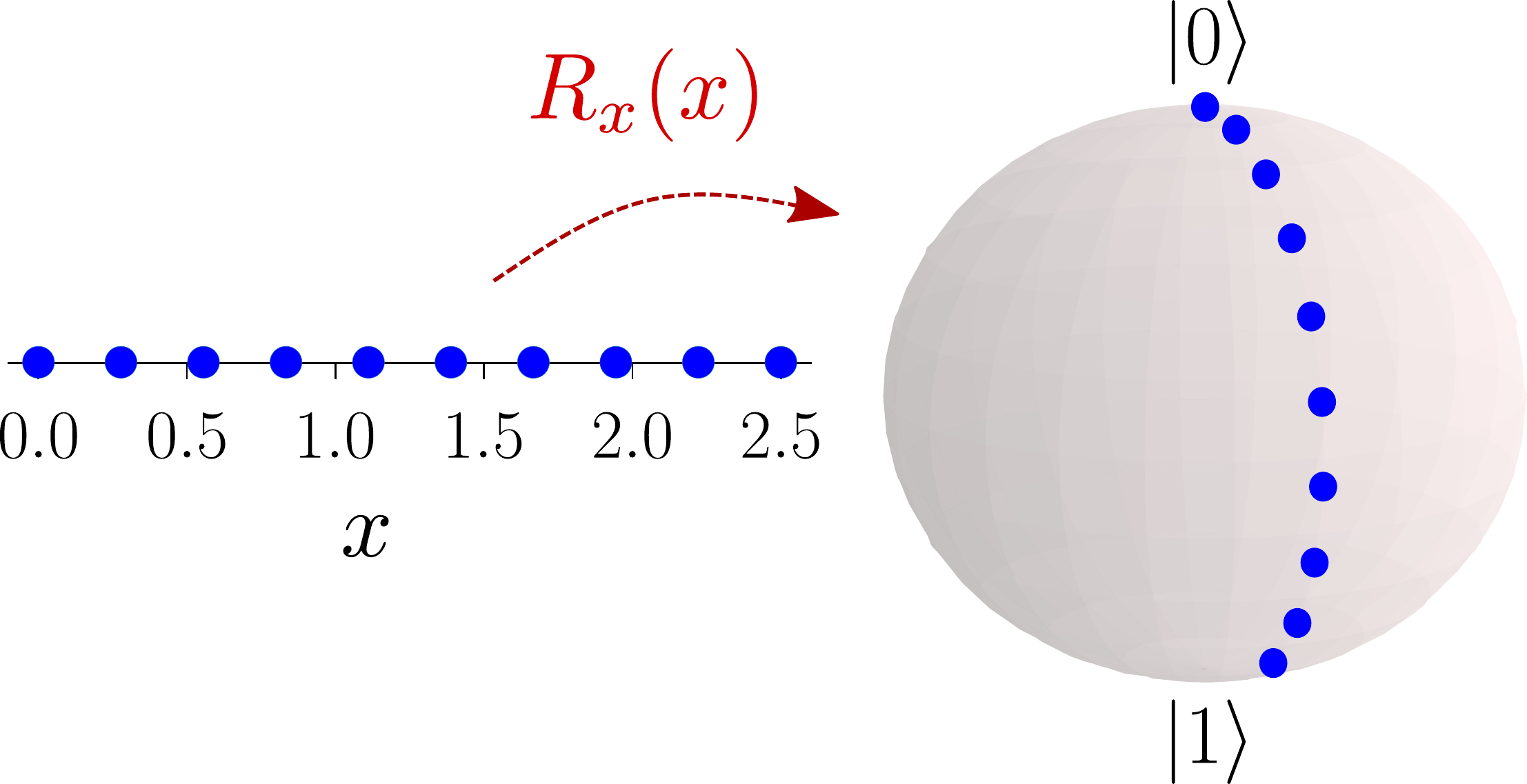}~~~~~~~~~\includegraphics[width = 0.2\textwidth]{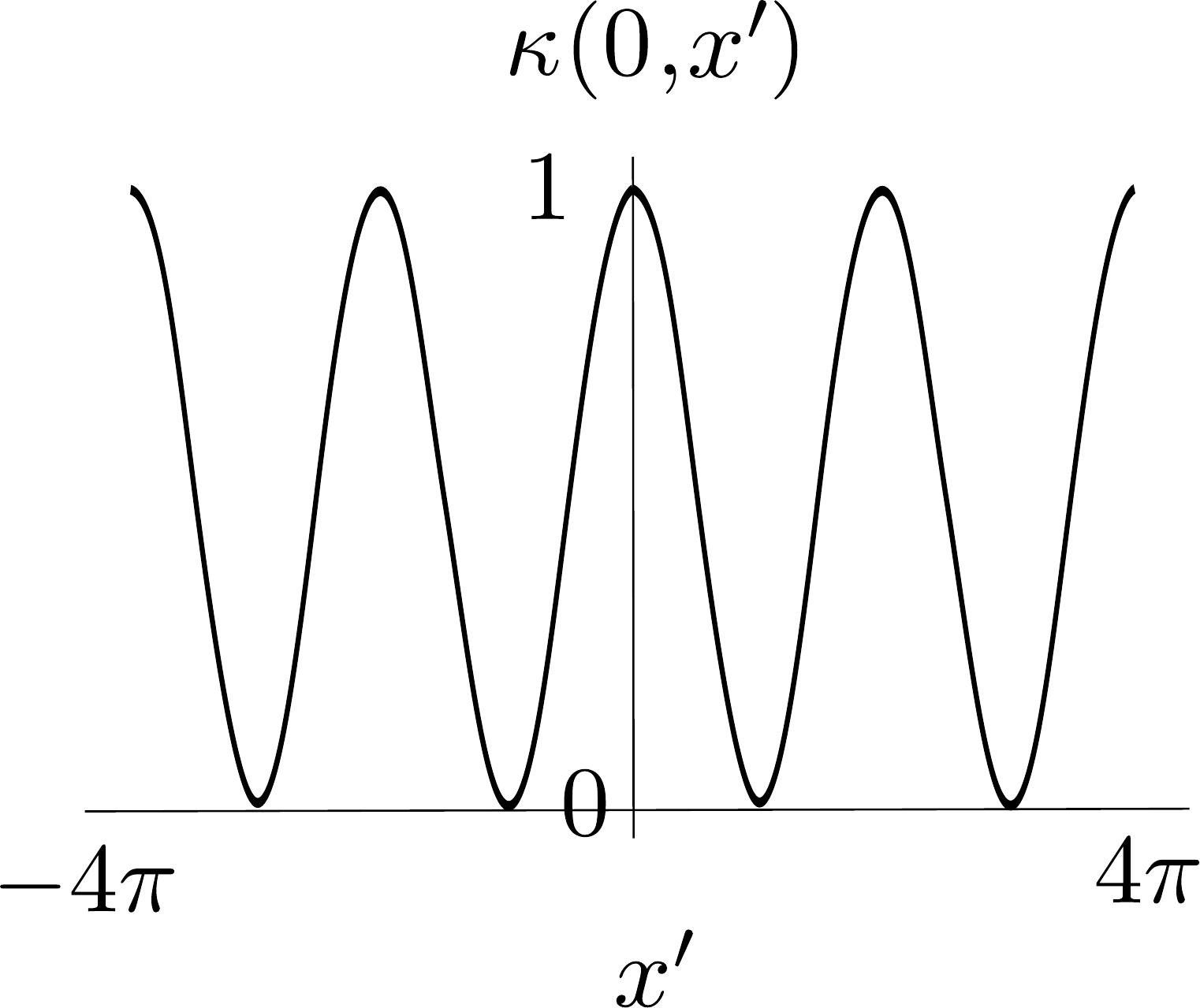}
    \caption{\textbf{Example of a data-encoding feature map and quantum kernel.} A scalar input is encoded into a single-qubit quantum state, which is represented as a point on a Bloch sphere. The embedding uses a feature map facilitated by a Pauli-X rotation. As can be seen when plotting the quantum states encoding equidistant points on an interval, the embedding preserves the structure of the data rather well, but is periodic. The embedding gives rise to a quantum kernel $\kappa$. When we fix the first input at zero, we can visualise the distance measure, which is a squared cosine function.}
    \label{fig:encoding_example}
\end{figure}

Consider an embedding that encodes a scalar input $x\in \mathbb{R}$ into the quantum state of a single qubit. The embedding is implemented by the Pauli-X rotation gate $R_X(x) = e^{- i \frac{x}{2} \sigma_x}$, where $\sigma_x$ is the Pauli-X operator. The data-encoding feature map is then given by $\phi: x \to \rho(x)$ with
\begin{equation}
\rho(x)  = \cos^2\left(\frac{x}{2}\right) \ketbra{0}{0} 
+ i \cos\left(\frac{x}{2}\right)\sin\left(\frac{x}{2}\right)  \ketbra{0}{1} 
-i \cos\left(\frac{x}{2}\right)\sin\left(\frac{x}{2}\right)  \ketbra{1}{0} 
+ \sin^2\left(\frac{x}{2}\right)  \ketbra{1}{1} ,
\end{equation}
and the quantum kernel becomes
\begin{equation}
    \kappa(x, x') = \left|\cos \left(\frac{x}{2}\right)\cos \left(\frac{x'}{2}\right)+ \sin \left(\frac{x}{2}\right)\sin \left(\frac{x'}{2}\right) \right|^2 =  \cos\left(\frac{x-x'}{2}\right)^2,
\end{equation}
which is a translation invariant squared cosine kernel. We will stick with this simple example throughout the following sections. It is illustrated in Figure~\ref{fig:encoding_example}.
\end{example}

\subsection{Making sense of matrix-valued feature vectors}\label{ssec:vectorised}

For readers that struggle to think of density matrices as feature vectors the data-encoding feature map (and further below, linear models) may be hard to visualise. I want to therefore insert a brief comment on an alternative version of the data-encoding feature map.

For all matters and purposes, the data-encoding feature map can be replaced by an alternative formulation
\begin{equation}
    \phi_v: \mathcal{X} \rightarrow \mathcal{F}_v \subset \mathcal{H} \otimes \mathcal{H}^*, 
\end{equation}
\begin{equation} 
    \phi_v =  \ket{\phi(x)} \otimes \ket{\phi^*(x)},
\end{equation}
where $\ket{\phi^*(x)}$ denotes the quantum state created from applying the complex conjugated (but not transposed) unitary $\ket{\phi^*(x)} = U^*(x) \ket{0}$ instead of $\ket{\phi(x)} = U(x) \ket{0}$, and $\mathcal{F}_v$ is the space of tensor products of a data-encoding Dirac vector with its complex conjugate. Note that since the complex conjugate of a unitary is a unitary, the unusual notation $\ket{\phi^*(x)}$ describes a valid quantum state which can be prepared by a physical circuit. The alternative feature space $\mathcal{F}_v$ is a subspace of the Hilbert space $\mathcal{H} \otimes \mathcal{H}^*$ with the property that inner products are real. One can show (but I won't do it here) that $\mathcal{F}_v$ is indeed a Hilbert space.

The inner product in this alternative feature space $\mathcal{F}_v$ is the absolute square of the inner product in the Hilbert space $\mathcal{H}$ of quantum states,
\begin{equation} 
\langle \psi | \varphi \rangle_{\mathcal{F}_v} =  |\braket{\psi}{\varphi}_{\mathcal{H}}|^2, 
\end{equation} 
and is therefore equivalent to the inner product in $\mathcal{F}$. This guarantees that it leads to the same quantum kernel.

The subscript $v$ refers to the fact that  $\ket{\phi(x)} \otimes \ket{\phi^*(x)}$ is a \textit{vectorisation} of $\rho(x)$, which reorders the $2^n$ matrix elements as a vector in $\mathbb{C}^{4n}$. 
To see this, let us revisit Example~\ref{ex:quantum_kernel} from above.
\begin{example}\label{ex:enc_kernel}
Consider the embedding from Example~\ref{ex:quantum_kernel}. The vectorised version of the data-encoding feature map is given by 
\begin{align}
    \phi_v : x \to \ket{\phi(x)} \otimes \ket{\phi^*(x)} &= \left(\cos \left(\frac{x}{2}\right) \ket{0} - i \sin\left(\frac{x}{2}\right) \ket{1}\right) \otimes \left(\cos \left(\frac{x}{2}\right) \ket{0} + i \sin\left(\frac{x}{2}\right) \ket{1}\right) \\
    &= 
    \begin{pmatrix}
      \cos^2\left(\frac{x}{2}\right)\\
      i \cos\left(\frac{x}{2}\right)\sin\left(\frac{x}{2}\right)\\
      -i \cos\left(\frac{x}{2}\right)\sin\left(\frac{x}{2}\right)\\
      \sin^2\left(\frac{x}{2}\right)
    \end{pmatrix},
\end{align}
and one can verify easily that the inner product of two such vectors leads to the same kernel.
\end{example}
Vectorised density matrices are common in the theory of open quantum systems \cite{jagadish2019invitation}, where they are written as $\kett{ \rho}$ (see also the \textit{Choi-Jamiolkowski isomorphism}). I will adopt this notation in Section \ref{ssec:meas_expansion} below to replace the Hilbert-Schmidt inner product $\tr{\rho^{\dagger} \sigma}$ with $\braakett{\rho }{ \sigma }$, which can be more illustrative at times. Note that the vectorised feature map, as opposed to Definition~\ref{def:data_enc_feat_map}, cannot capture mixed quantum states and is therefore less powerful.

\section{Examples of quantum kernels}\label{ssec:examples_kernel}

\begin{table}[t]
    \centering
    \def\arraystretch{1.5}
    \begin{tabular}{l l}
        \hline \hline
        \textbf{encoding} & \textbf{kernel $\kappa(x, x')$}\\ 
         basis encoding & $ \delta_{x, x}$ \\
         amplitude encoding & $ |\mathbf{x}^{\dagger}\mathbf{x}'|^2$\\
         repeated amplitude encoding & $ (|\mathbf{x}^{\dagger}\mathbf{x}'|^2)^r$\\
         rotation encoding & $\prod_{k=1}^N |\cos (x'_k - x_k)|^2$\\
         coherent state encoding & $ e^{- |\mathbf{x} - \mathbf{x}'|^2 }$\\
         general near-term encoding & $   \sum_{s, t \in \Omega} e^{i \mathbf{s} \mathbf{x}}e^{i \mathbf{t}\mathbf{x}'} c_{\mathbf{s}, \mathbf{t}}$\\
        \hline \hline
    \end{tabular}
    \caption{\textbf{Overview of data encoding strategies used in the literature and their quantum kernels.} If bold notation is used, the input domain is assumed to be the $\mathcal{X} \subseteq \mathbb{R}^N$.}
    \label{tab:encoding_kernels}
\end{table}

To fill the definition of the quantum kernel with life, let us have a look at typical information encoding strategies or data embeddings in quantum machine learning, and the kernels they give rise to (following \cite{schuld2019quantum}, and see Table~\ref{tab:encoding_kernels}). Note that it has been shown that there are kernels that cannot be efficiently computed on classical computers \citep{liu2020rigorous}.\footnote{ The argument basically defines a feature map based on a computation that is conjectured by quantum computing research to be classically hard. } As important as such results are, the question of quantum kernels that are actual useful for every-day problems is still wide open.

\subsection{Data encoding that relates to classical kernels}

The following strategies to encode data all have resemblance to kernels from the classical machine learning literature. This means that, sometimes up to an absolute square value, we can identify them with standard kernels such as the polynomial or Gaussian kernel. These kernels are plotted in Figure~\ref{fig:kernels_standard} using simulations of quantum computations implemented in the quantum machine learning software library PennyLane \citep{bergholm2018pennylane}. Note that I switch to bold notation when the input space is $\mathbb{C}^N$ or $\mathbb{R}^N$ \\

\textbf{Basis encoding.} Basis encoding is possibly the most common information encoding strategy in qubit-based quantum computing. Inputs $x \in \mathcal{X}$ are assumed to be binary strings of length $n$, and $\mathcal{X} = \{0, 1\}^{\otimes n}$. Every binary string has a unique integer representation $i_x = \sum_{k=0}^{n-1} 2^k x_k$. The data-encoding feature map maps the binary string to a computational basis state,
\begin{equation}\phi: x \rightarrow \ketbra{i_x}{i_x}. \end{equation}
The quantum kernel is given by the Kronecker delta
\begin{equation} \kappa(x,x') = |\braket{i_{x'}}{j_{x}}|^2 = \delta_{x,x'}, \end{equation}
which is of course a very strict similarity measure on input space, and arguably not the best choice of data encoding for quantum machine learning tasks. Basis encoding requires $\mathcal{O}(n)$ qubits.\\

\textbf{Amplitude encoding.} Amplitude encoding assumes that $\mathcal{X} = \mathbb{C}^{2^n}$, and that the inputs are normalised as $\|\mathbf{x}\|^2 = \sum_i |x_i|^2 = 1$. The data-encoding feature map associates each input with a quantum state whose amplitudes in the computational basis are the elements in the input vector,
\begin{equation}\phi: \mathbf{x} \rightarrow  \ketbra{\mathbf{x}}{\mathbf{x}} =  \sum_{i,j=1}^N x_i x^*_j \ketbra{i}{j}. \end{equation}
This data-encoding strategy leads to an identity feature map, which can be implemented by a non-trivial quantum circuit (for obvious reasons also known as ``arbitrary state preparation''), which takes time $\mathcal{O}(2^n)$ \citep{iten2016quantum}. The quantum kernel is the absolute square of the linear kernel
\begin{equation} \kappa(\mathbf{x},\mathbf{x}') = |\braket{\mathbf{x}'}{\mathbf{x}}|^2 = |\mathbf{x}^{\dagger}\mathbf{x}'|^2. \end{equation}
It is obvious that this quantum kernel does not add much power to a linear model in the original feature space, and it is more of interest for theoretical investigations that want to eliminate the effect of the feature map. Amplitude encoding requires $\mathcal{O}(n)$ qubits.\\
 
\textbf{Repeated amplitude encoding.} Amplitude encoding can be repeated $r$ times, 
\begin{equation} \phi: \mathbf{x} \rightarrow \ketbra{\mathbf{x}}{\mathbf{x}}  \otimes \cdots \otimes \ketbra{\mathbf{x}}{\mathbf{x}}  \end{equation}
to get powers of the quantum kernel in amplitude encoding
\begin{equation} \kappa(\mathbf{x},\mathbf{x}') = (|\braket{\mathbf{x}'}{\mathbf{x}}|^2 )^r = (|(\mathbf{x}')^{\dagger}\mathbf{x}|^2)^r. \end{equation}
A constant non-homogenity can be added by extending the original input with constant dummy features. Repeated amplitude encoding requires $\mathcal{O}(r n)$ qubits.\\

\textbf{Rotation encoding.} Rotation encoding is a qubit-based embedding that assumes $\mathcal{X} = \mathbb{R}^{n}$ (where $n$ is again the number of qubits) without any normalisation condition. Since it is $2\pi$-periodic one may want to limit $\mathbb{R}^{n}$ to the hypercube  $[0, 2\pi]^{\otimes n}$. The $i$th feature $x_i$ is encoded into the $i$th qubit via a Pauli rotation. For example, a Pauli-Y rotation puts the qubit into state $\ket{q_i(x_i)} = \cos(x_i) \ket{0} + \sin(x_i) \ket{1}$. The data-encoding feature map is therefore given by
\begin{equation}\phi: \mathbf{x} \rightarrow \ketbra{\phi(\mathbf{x})}{\phi(\mathbf{x})} \text{ with } \ket{\phi( \mathbf{x})} = \sum_{q_1,\dots, q_n = 0}^{1} \prod_{k=1}^n \cos (x_k)^{q_k}  \sin (x_k)^{1-q_k} \ket{q_1, \dots, q_n}, \end{equation}
and the corresponding quantum kernel is related to the cosine kernel:
\begin{equation} \kappa(\mathbf{x},\mathbf{x}') = \prod_{k=1}^n  |\sin x_k \sin x'_k + \cos x_k \cos x'_k|^2 = \prod_{k=1}^n |\cos (x_k - x'_k)|^2. \end{equation}
Rotation encoding requires $\mathcal{O}(n)$ qubits.\\

\textbf{Coherent state encoding.} Coherent states are known in the field of quantum optics as a description of light modes. Formally, they are superpositions of so called \textit{Fock states}, which are basis states from an infinite-dimensional discrete basis $\{\ket{0}, \ket{1}, \ket{2},...\}$, instead of the binary basis of qubits. 
A coherent state has the form
\begin{equation} \ket{\alpha} = e^{- \frac{|\alpha|^2}{2}} \sum\limits_{k=0}^{\infty} \frac{\alpha^k}{\sqrt{k!}} \ket{k},\end{equation}
for $\alpha \in \mathbb{C}$.
Encoding a real scalar input $x_i \in \mathbb{R}$ into a coherent state $\ket{\alpha_{x_i}}$, corresponds to a data-encoding feature map with an infinite-dimensional feature space, 
\begin{equation}\phi: x_i   \rightarrow \ketbra{\alpha_{x_i}}{\alpha_{x_i}}, \text{ with } \ket{\alpha_{x_i}} = e^{- \frac{|x_i|^2}{2}} \sum\limits_{k=0}^{\infty} \frac{x_i^k}{\sqrt{k!}} \ket{k} .  \end{equation}
We can encode a real vector $\mathbf{x} =(x_1,...,x_n)$ into $n$ joint coherent states,
\begin{equation}  \ketbra{\alpha_{\mathbf{x}}}{\alpha_{\mathbf{x}}} = \ketbra{\alpha_{x_1}}{\alpha_{x_1}} \otimes  \dots \otimes  \ketbra{\alpha_{x_n}}{\alpha_{x_n}}. \end{equation}
The quantum kernel is a Gaussian kernel \cite{chatterjee2016generalized}:
\begin{equation}\kappa(\mathbf{x}, \mathbf{x}') =  \left|e^{- \left( \frac{ |\mathbf{x}|^2}{2} +\frac{|\mathbf{x}'|^2}{2} - \mathbf{x}^{T}\mathbf{x}'   \right)}\right|^2 = e^{- |\mathbf{x} - \mathbf{x}'|^2 } \end{equation}
Preparing coherent states can be done with displacement operations in quantum photonics.

\begin{figure}
    \centering
    \includegraphics[width=\textwidth]{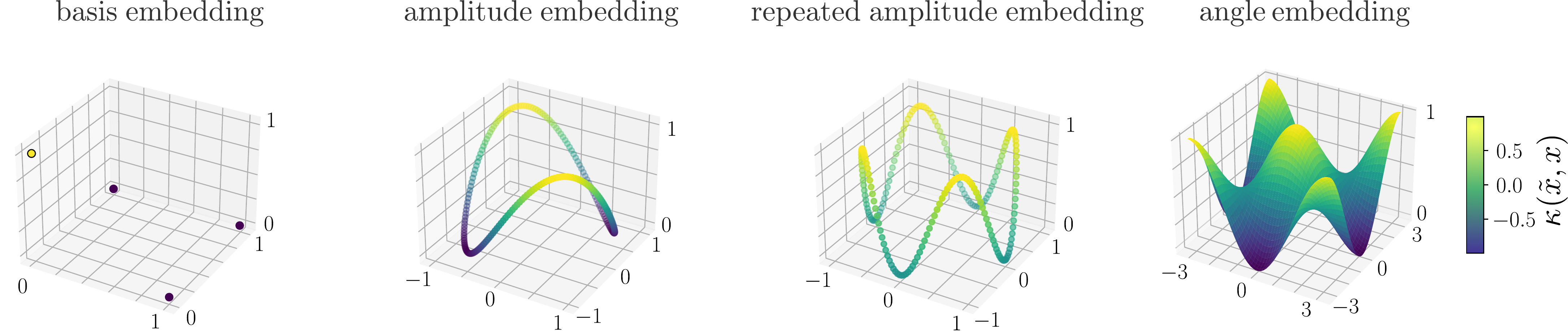}
    \caption{\textbf{Quantum kernels of different data embeddings.} Plots of some of the functions $\kappa(\tilde{x}, x)$ for the kernels introduced above, using $\mathbf{x} = (x_1, x_2) \in \mathbb{R}^2$ for illustration purposes. The first entry $\tilde{\mathbf{x}}$ is fixed at $\tilde{\mathbf{x}} = (0,0)$ for basis and rotation embedding, and at $\tilde{\mathbf{x}}= (\frac{1}{\sqrt{2}}, \frac{1}{\sqrt{2}})$ for the variations of amplitude embedding. The second value is depicted as the x-y plane. }
    \label{fig:kernels_standard}
\end{figure}

\subsection{Fourier representation of the quantum kernel}

It is suspicious that all embeddings plotted in Figure~\ref{fig:kernels_standard} have a periodic, trigonometric structure. This is a fundamental characteristic of how physical parameters enter quantum states. To see this we will define a general class of embeddings (also called ``time-evolution encoding'') that is used a lot in near-term quantum machine learning, and which includes all examples above if we allow for classical pre-processing of the features. This strategy assumes that $\mathcal{X} = \mathbb{R}^N$ for some arbitrary $N$ (whose relation to the number of qubits $n$ depends on the embedding), which means that I will stick to the bold notation. The embedding of $x_i$ is executed by gates of the form $e^{-i x_i G_i }$ where $G_i$ is $d_i \leq 2^n$-dimensional Hermitian operator called the \textit{generating Hamiltonian}. For the popular choice of Pauli rotations, $G_i = \frac{1}{2} \sigma$ with the Pauli operator $\sigma \in \{\sigma_z, \sigma_y, \sigma_z\}$. The gates can be applied to different qubits as in rotation encoding, or to the same qubits, and to be general we allow for arbitrary quantum computations between each encoding gate. 

Refs. \cite{vidal2019input} and \cite{schuld2020effect} showed that the Dirac vectors $\ket{\phi( \mathbf{x})}$ can be represented in terms of periodic functions of the form $e^{i x_i \omega}$, where $\omega \in \mathbb{R}$ can be interpreted as a frequency. The frequencies involved in the construction of the data-encoding feature vectors are solely determined by the generating Hamiltonians $\{G_i\}$ of the gates that encode the data. For popular choices of Hamiltonians, the frequencies $\omega$ are integer-valued, which means that the feature space is constructed from \textit{Fourier basis functions} $e^{i x_i n}, n \in \mathbb{Z}$. This allows us to describe and analyse the quantum kernel with the tools of Fourier analysis.

Let me state the result for the simplified case that each input $x_i$ is only encoded once, and that all the encoding Hamiltonians are the same ($G_1 = \dots = G_N = G$). The proof is deferred to Appendix~\ref{app:fourier_proof}, which also shows how our example of Pauli-X encoding can be cast as a Fourier series.

\begin{theorem}[Fourier representation of the quantum kernel]\label{thm:kernel_fourier}
Let $\mathcal{X} = \mathbb{R}^N$ and $S(\mathbf{x})$ be a quantum circuit that encodes the data inputs $\mathbf{x} = (x_1, \dots, x_N) \in \mathcal{X}$ into a $n$-qubit quantum state $S(\mathbf{x}) \ket{0} = \ket{\phi(\mathbf{x})}$ via gates of the form $e^{-i x_i G }$ for $i = 1,\dots, N$. Without loss of generality $G$ is assumed to be a $d \leq 2^n$-dimensional diagonal operator with spectrum $\lambda_1, \dots, \lambda_d$. Between such data-encoding gates, and before and after the entire encoding circuit, arbitrary unitary evolutions $W^{(1)}, \dots, W^{(N+1)}$ can be applied, so that 
\begin{equation}
    S(\mathbf{x}) = W^{(N+1)} e^{-i x_N G} W^{(N)} \dots W^{(2)} e^{-i x_1 G} W^{(1)}.
\end{equation}
The quantum kernel $\kappa(\mathbf{x}, \mathbf{x}')$ can  be written as
\begin{equation}\label{eq:kernel_fourier}
    \kappa(\mathbf{x}, \mathbf{x}') = \sum_{\mathbf{s}, \mathbf{t} \in \Omega} e^{i \mathbf{s} \mathbf{x}} e^{i \mathbf{t} \mathbf{x}'} c_{\mathbf{st}},  
\end{equation}
where $\Omega \subseteq \mathbb{R}^N$, and $c_{\mathbf{st}} \in \mathbb{C}$.
For every $\mathbf{s},\mathbf{t} \in \Omega$ we have $-\mathbf{s},-\mathbf{t} \in \Omega$ and $c_{\mathbf{st}} = c^*_{-\mathbf{s} -\mathbf{t}}$, which guarantees that the quantum kernel is real-valued. 
\end{theorem}

While the conditions of this theorem may sound restrictive at first, it includes a fairly general class of quantum models. The standard way to control a quantum system is to apply an evolution of Hamiltonian $G$ for time $t$, which is exactly described by the form $e^{-i t G}$. The time $t$ is associated with the input to the quantum computer (which may be the original input $x \in \mathcal{X}$ or the result of some pre-processing, in which case we can just redefine the dataset to be the pre-processed one). In short, most quantum kernels will be of the form shown in Eq.~(\ref{eq:kernel_fourier}). 

Importantly, for the class of Pauli generators, the kernel becomes a Fourier series:

\begin{corollary}[Fourier series representation of the quantum kernel]
For the setting described in Theorem~\ref{thm:kernel_fourier}, if the eigenvalue spectrum of $G$ is such that any difference $\lambda_i - \lambda_j$ for $i, j = 1, \dots, d$ is in $\mathbb{Z}$, then $\Omega$ becomes the set of $N$-dimensional integer-valued vectors $\mathbf{n} = (n_1, \dots, n_N)$, $n_1,\dots n_N \in \mathbb{Z}$. In this case the quantum kernel is a multi-dimensional Fourier series,
\begin{equation}\label{eq:kernel_fourier_int}
    \kappa(\mathbf{x}, \mathbf{x}') = \sum_{\mathbf{n}, \mathbf{n'} \in \Omega} e^{i \mathbf{n} \mathbf{x}} e^{i \mathbf{n'} \mathbf{x}'} c_{\mathbf{n, n'}},  
\end{equation}
\end{corollary}

\begin{figure}
    \centering
    \includegraphics[width=0.75\textwidth]{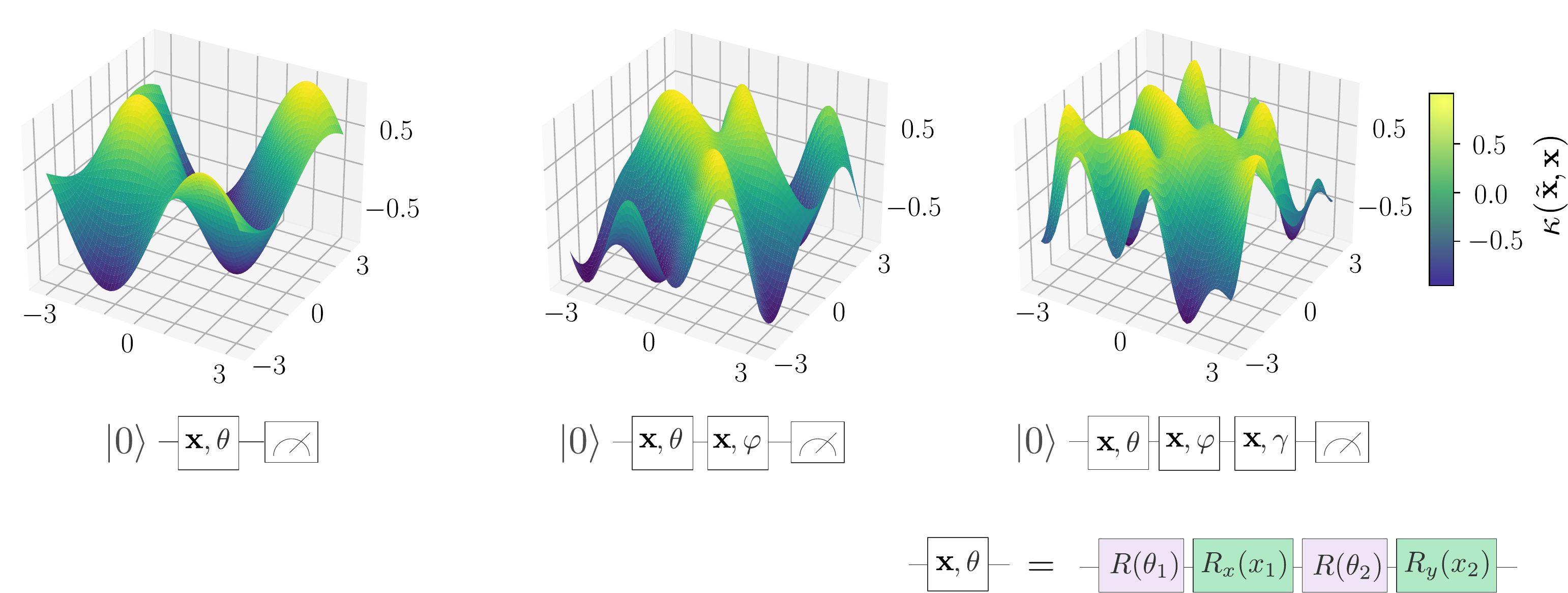}
    \caption{\textbf{Kernels generated by rotation embeddings.} Plots of the quantum kernel $\kappa(\tilde{\mathbf{x}}, \mathbf{x})$ with  $\tilde{\mathbf{x}} = (0, 0)$ using a very general data encoding strategy that repeats the input encoding into a single qubit one, two and three times. It is obvious that the repetition decreases the smoothness of the kernel by increasing the Fourier basis functions from which the kernel is inherently constructed. }
    \label{fig:kernels_nearterm}
\end{figure}

Expressions~(\ref{eq:kernel_fourier}) and ~(\ref{eq:kernel_fourier_int}) reveal a lot about the structure of quantum kernels, for example that they are not necessarily translation invariant, $\kappa(\mathbf{x}, \mathbf{x}') \neq g(\mathbf{x}-\mathbf{x}')$, unless the data-encoding strategy leads to $c_{\mathbf{st}} =\tilde{c}_{\mathbf{st}} \delta_{\mathbf{st}} = c_{\mathbf{s}}$ and
\begin{equation}
    \kappa(\mathbf{x}, \mathbf{x}') = \sum_{\mathbf{s} \in \Omega} e^{i \mathbf{s} (\mathbf{x}-\mathbf{x}')} \tilde{c}_{\mathbf{s}}.  
\end{equation}
Since $e^{-i x_i G} e^{i x'_i G}  = e^{-i (x_i - x'_i) G}$, this is true for all data embeddings that encode each original input into a separate physical subsystem, like rotation encoding introduced above. 

It is an interesting question if this link between data embedding and Fourier basis functions given to us by physics can help design particularly suitable kernels for applications, or be used to control smoothness properties of the kernel in a useful manner.

\section{Quantum models and reproducing kernel Hilbert spaces}\label{sec:models_rkhs}

I will now discuss the observation that quantum models are linear models in the feature space $\mathcal{F}$ of the data-encoding feature map. This automatically allows us to apply the results of kernel methods to quantum machine learning. A beautiful summary of these results can be found in \cite{steinwart2008support} and \cite{scholkopf2002learning}, which serve as a basis for many of the following insights.

\subsection{Quantum models are linear models in feature space}\label{ssec:linear}

First, let us define a quantum model. For this we need measurements.

\begin{quote}
    \textit{Measurements.}  In quantum computing, a measurement produces the observable result of a quantum circuit, and can therefore be seen as the final step of a quantum algorithm\footnote{An important exception is when the outcome of a measurement is used to influence the quantum circuit itself, but I do not consider those complications here.}. Mathematically speaking, a measurement corresponds to a Hermitian operator $\meas$ acting on vectors in the Hilbert space of the quantum system $\mathcal{H}$. Just like density matrices, measurement operators can be represented as elements of the space of $2^n \times 2^n$-dimensional complex matrices \cite{wolf}, and therefore live in a subspace of the data-encoding feature space $\mathcal{F}$. This will become quite crucial below. 
    
    A Hermitian operator can always be diagonalised and written as 
    \begin{equation}\meas = \sum_i \mu_i \ketbra{\mu_i}{\mu_i},\end{equation}
    where $\mu_i$ are the eigenvalues of $\meas$ and $\{\ket{\mu_i}\}$ is an orthonormal basis in the Hilbert space $\mathcal{H}$ of the quantum system. Note that $\ketbra{\mu_i}{\mu_i}$ is an outer product, and can be thought of as a (density) matrix. 
    
    The apparatus of quantum theory allows us to compute expected outcomes or \textit{expectations} of measurement results. Such expectations derive from expressing the quantum state in the eigenbasis of the measurement operator, $\ket{\psi} = \sum_i \braket{\mu_i}{\psi} \ket{\mu_i}$, and using the fact that $\meas \ket{\mu_i} = \mu_i \ket{\mu_i}$ and $\braket{\mu_i}{\mu_i}=1$: 
    \begin{equation} \label{eq:exp}
    \tr{\rho \meas } = \langle \psi | \meas | \psi \rangle = \sum_{i,j} \braket{\psi}{\mu_j} \braket{\mu_i}{\psi} \bra{\mu_j} \meas   \ket{\mu_i}  =  \sum_{i} |\braket{\psi}{\mu_i}|^2 \mu_i = \sum_{i}p(\mu_i) \mu_i.
    \end{equation}
    The above used the ``Born rule'', which states that the probability of measuring outcome $\mu_i$ is given by 
    \begin{equation} \label{eq:probs}
        p(\mu_i) = |\langle \mu_i | \psi \rangle|^2.
    \end{equation}
    It is clear that the right hand side of Eq.~(\ref{eq:exp}) is an expectation of a random variable in the classical sense of probability theory, but the probabilities themselves are computed by an unusual mathematical framework. Finally, it is good to know that the expectation of a measurement $\meas_{\varphi} = \ketbra{\varphi}{\varphi}$ (where $\ket{\varphi}$ is an arbitrary quantum state) gives us the \textit{overlap} of $\ket{\varphi}$ and $\ket{\psi}$, 
    \begin{equation} \label{eq:overlap}
        \tr{\rho \meas_{\varphi}} = \langle \psi | \meas_{\varphi} | \psi \rangle = |\langle \varphi | \psi \rangle|^2.
    \end{equation}
    Note that only because we can write down a measurement mathematically, we cannot necessarily implement it efficiently on a quantum computer. However, for measurements of type $\meas_{\varphi}$ there is a very efficient routine called the SWAP test to do so, if we can prepare the corresponding state efficiently. In practice, more complicated measurements are implemented by applying a circuit $W$ to the final quantum state, followed by a simple measurement (such as the well-known Pauli-Z measurement $\sigma_z$ that probes the state of qubits, which effectively implements $\meas = W^{\dagger} \sigma_z W$). 
    
    Of course, actual quantum computers can only ever produce an estimate of the above statistical properties, namely by repeating the entire computation $K$ times and computing the empirical probability/frequency or the empirical expectation $\frac{1}{K} \sum_{i=1}^K \mu_i$. However, repeating a fixed computation tens of thousands of times can be done in a fraction of a second on most hardware platforms, and only leads to a small constant overhead.
\end{quote}

We can define a quantum model as a measurement performed on a data-encoding state:
\begin{definition}[Quantum model]
    Let $\rho(x)$ be a quantum state that encodes classical data $x \in \mathcal{X}$ and $\meas$ a Hermitian operator representing a quantum measurement.
    A quantum model is the expectation of the quantum measurement as a function of the data input,
    \begin{equation}\label{eq:quantum_model}
        f(x) = \tr{ \rho(x) \meas }.
    \end{equation}
    The space of all quantum models contains functions $f: \mathcal{X} \rightarrow \mathbb{R}$.
    For pure-state embeddings with $\rho(x)= \ketbra{\phi(x)}{\phi(x)}$, this simplifies to
    \begin{equation}
    f(x) = \bra{\phi(x)} \meas \ket{\phi(x)}.
    \end{equation}
\end{definition}
As mentioned above, this definition is very general, but does not consider the important class of generative quantum models.

\begin{example}
   Getting back to the standard example of the Pauli-X rotation encoding, we can upgrade it to a full quantum model with parametrised measurement by applying an additional arbitrary rotation $R(\theta_1, \theta_2, \theta_3)$, which is parametrised by three trainable angles and is expressive enough to represent any single-qubit computation. After this, we measure in the Pauli-Z basis, yielding the overall quantum model:
   \begin{equation}
       f(x) = \tr{ \rho(x)  \meas(\theta_1, \theta_2, \theta_3)}=  \bra{\phi(x)} \meas(\theta_1, \theta_2, \theta_3)\ket{\phi(x)},
   \end{equation}
   with measurement $\meas(\theta_1, \theta_2, \theta_3) = R^{\dagger}(\theta_1, \theta_2, \theta_3) \sigma_z  R(\theta_1, \theta_2, \theta_3)$,
\begin{equation}
 	R(\theta_1, \theta_2, \theta_3) = \begin{pmatrix}
    e^{i(-\frac{\theta_1}{2}-\frac{\theta_3}{2})}\cos(\frac{\theta_2}{2}) & -e^{i(-\frac{\theta_1}{2}+\frac{\theta_3}{2})}\sin(\frac{\theta_2}{2})\\
    e^{i(\frac{\theta_1}{2}-\frac{\theta_3}{2})}\sin(\frac{\theta_2}{2}) &
    e^{i(\frac{\theta_1}{2}+\frac{\theta_3}{2})}\cos(\frac{\theta_2}{2})
    \end{pmatrix}
\end{equation}   
    and $\ket{\phi(x)} =  R_x(x) \ket{0}$.  One can use a computer-algebra system (or, for the patient among us, lengthy calculations) to verify that the quantum model is equivalent to the function
   \begin{equation}\label{eq:example_model}
   f(x) =  \cos (\theta_2) \cos (x) - \sin(\theta_1) \sin (\theta_2) \sin (x),
   \end{equation}
   and hence independent of the third parameter. 
\end{example}

Next, let us define what a linear (machine learning) model in feature space is:

\begin{definition}[Linear model]
    Let $\mathcal{X}$ be a data domain and $\phi:\mathcal{X} \to \mathcal{F}$ a feature map. We call any function
    \begin{equation}
        f(x) = \langle \phi(x), w \rangle_{\mathcal{F}},
    \end{equation}
    with $w \in \mathcal{F}$ a linear model in $\mathcal{F}$.
\end{definition}

From these two definitions we immediately see that:
\begin{theorem}[Quantum models are linear models in data-encoding feature space] \label{thm:quantum_linear}
    Let $f(x) = \tr{\rho \meas}$ be a quantum model with feature map $\phi: x \in \mathcal{X} \to \rho(x) \in \mathcal{F}$ and data domain $\mathcal{X}$. The quantum model $f$ is a linear model in $\mathcal{F}$.
\end{theorem}

It is interesting to note that the measurement $\meas$ can always be expressed as a linear combination $\sum_k \gamma_k \rho(x^k)$ of data-encoding states $\rho(x^k)$ where $x^k \in \mathcal{X}$. 

\begin{theorem}[Quantum measurements are linear combinations of data-encoding states] \label{thm:measurement_expansion_general}
    Let $f_{\meas}(x) = \tr{\rho \meas}$ be a quantum model. There exists a measurement $ \meas_{\rm exp} \in \mathcal{F}$ of the form
    \begin{equation} \meas_{\rm exp} = \sum_k \gamma_k \rho(x^k)\end{equation}
    with $ x^k \in \mathcal{X}$, such that $f_{\meas}(x) = f_{\meas_{\rm exp}}(x)$ for all $x \in \mathcal{X}$. 
\end{theorem}

\begin{proof}
    We can divide $\meas$ into the part that lies in the image of $\mathcal{X}$ and the remainder $R$,
    \begin{equation}
        \meas = \meas_{\rm exp} + R.
    \end{equation}
    Since the trace is linear, we have:
    \begin{equation}
        \tr{\rho(x) \meas} = \tr{\rho(x)  \meas_{\rm exp}} + \tr{\rho(x) R}.
    \end{equation}
    The data-encoding state $\rho(x)$ only has contributions in $\mathcal{F}$, which means that the inner product $\tr{\rho(x) R}$ is always zero.
\end{proof}
Below we will see that optimal measurements with respect to typical machine learning cost functions can be expanded in the training data only.

Note that the fact that a quantum model can be expressed as a linear model in the feature space does \textit{not} mean that it is linear in the Hilbert space of the Dirac vectors $\ket{\phi(x)}$, nor is it linear in the data input $x$. As mentioned before, in the context of variational circuits the measurement usually depends on trainable parameters, which is realised by applying a parametrised quantum operation or circuit that ``rotates'' the basis of a fixed measurement. Variational quantum models are also not necessarily linear in their actual trainable parameters.

As a last comment for readers that prefer the vectorised version of the data-encoding feature map, by writing the measurement operator $\meas = \sum_i \mu_i \ketbra{\mu_i}{\mu_i}$ in its eigenbasis, we can likewise write a quantum model as the inner product of a vectorised feature vector $\ket{\phi(x)} \otimes \ket{\phi^*(x)} \in \mathcal{F}_v$ with some other vector $\sum_i \mu_i \ket{\mu_i} \otimes \ket{\mu_i} \in \mathcal{F}_v$. 
\begin{align} \label{eq:raw_model_general} 
    f(x) &= \bra{\phi(x)} \meas \ket{\phi(x)}\\
    &= \sum_i \mu_i |\braket{\mu_i}{\phi(x)}|^2\\
    &= \Big( \bra{\phi(x)} \otimes \bra{\phi^*(x)}\Big)  \Big( \sum_i \mu_i \ket{\mu_i} \otimes \ket{\mu^*_i}  \Big), \label{eq:raw_model} 
\end{align}
or using the vectorised density matrix notation introduced above,
\begin{equation}\label{eq:vectorised_linear_model}
    f(x) = \braakett{\rho(x)}{w},
\end{equation}
with $w = \sum_i \mu_i \kett{\rho_i} $.

\subsection{The RKHS of the quantum kernel and the space of quantum models are equivalent}\label{ssec:rkhs_quantum_models}

So far we were dealing with two different kinds of Hilbert spaces: The Hilbert space $\mathcal{H}$ of the quantum system, and the feature space $\mathcal{F}$ that contains the embedded data. I will now construct yet another feature space for the quantum kernel, but one derived directly from the kernel and with no further notion of a quantum model. This time the feature space is a Hilbert space $F$ of \textit{functions}, and due to its special construction it is called the \textit{reproducing kernel Hilbert space} (RKHS). The relevance of this feature space is that the functions it contains turn out to be exactly the quantum model functions $f$ (which is a bit surprising at first: this feature space contains linear models defined in an equivalent feature space!). 

The RKHS $F$ of the quantum kernel can be defined as follows (as per Moore-
Aronsajn's construction\footnote{See also \url{http://www.stats.ox.ac.uk/~sejdinov/teaching/atml14/Theory_2014.pdf} for a great overview.}):
\begin{definition}[Reproducing kernel Hilbert space]\label{def:rkhs}
Let $\mathcal{X} \neq \emptyset$. The reproducing kernel Hilbert space of a kernel $\kappa$ over $\mathcal{X}$ is the Hilbert space $F$ created by completing the span of functions $f: \mathcal{X} \rightarrow \mathbb{R}$, $f(\cdot) = \kappa(x, \cdot)$, $x \in \mathcal{X}$ (i.e., including the limits of Cauchy series). For two functions $f(\cdot) = \sum_i \alpha_i \kappa(x^i, \cdot)$, $g(\cdot)=\sum_j \beta_j \kappa(x^j, \cdot) \in F$, the inner product is defined as
\begin{equation}\langle f, g\rangle_{F} =  \sum_{ij} \alpha_i \beta_j \kappa(x^i, x^j),  \end{equation}
with $ \alpha_i, \beta_j \in \mathbb{R}$.
\end{definition}
Note that according to Theorem~\ref{thm:kernel_fourier} the ``size'' of the space of common quantum models, and likewise the RKHS of the quantum kernel, are fundamentally limited by the generators of the data-encoding gates. If we consider $\kappa$ as the quantum kernel, the definition of the inner product reveals with
\begin{equation}
    \langle \kappa(x, \cdot), \kappa(x', \cdot) \rangle_{F} =  \kappa(x, x'),
\end{equation}
that $x \rightarrow \kappa(x, \cdot)$ is a feature map of this kernel (but one mapping data to \textit{functions} instead of matrices, which feels a bit odd at first). In this sense, $F$ can be regarded as an alternative feature space to $\mathcal{F}$. The name of this unique feature space comes from the \textit{reproducing property} 
\begin{equation}
    \langle f, \kappa(x, \cdot) \rangle_{F} = f(x) \text{ for all } f \in F,
\end{equation}
which also shows that the kernel is the evaluation functional $\delta_x$ which assigns $f$ to $f(x)$. An alternative definition of the RKHS is the space in which the evaluation functional is bounded, which gives the space a lot of favourable properties from a mathematical perspective.

To most of us, the definition of an RKHS is terribly opaque when first encountered, so a few words of explanation may help (see also Figure~\ref{fig:rkhs}). One can think of the RKHS as a space whose elementary functions $\kappa(x, \cdot)$ assign a distance measure to every data point. Functions of this form were also plotted in Figure~\ref{fig:kernels_standard} and \ref{fig:kernels_nearterm}. By feeding another data point $x'$ into this ``similarity measure'', we get the distance between the two points. As a vector space, $F$ also contains linear combinations of these building blocks. The functions living in $F$ are therefore linear combinations of data similarities, just like for example kernel density estimation constructs a smooth function by adding Gaussians centered in the data. The kernel then regulates the ``resolution'' of the distance measure, for example by changing the variance of the Gaussian. 

\begin{figure*}[t]
    \centering
    \includegraphics[width=0.4\textwidth]{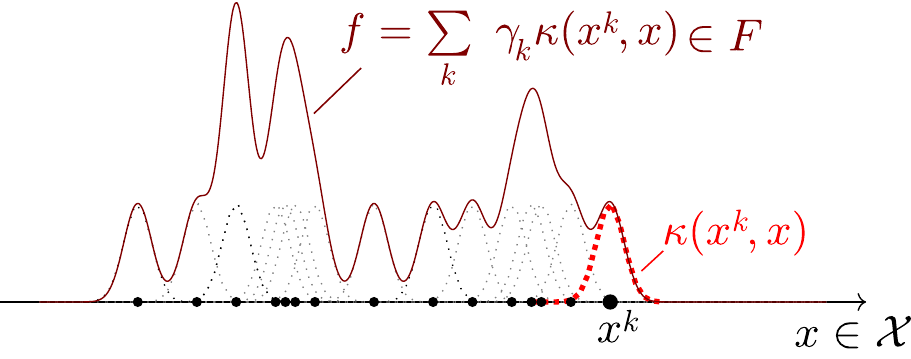}
    \caption{\textbf{Intuition for the functions living in the reproducing kernel Hilbert space (RKHS).} The RKHS $F$ contains functions that are linear combinations of kernel functions where one ``slot'' is fixed in a possible data sample $x^k \in \mathcal{X}$. This illustration of one such function $f \in F$, using a Gaussian kernel, shows how the kernel regulates the ``smoothness'' of the functions in $F$, as a wider kernel will simplify $f$. Since the RKHS is equivalent to the space of linear models that it has been derived from, the kernel fundamentally defines the class of functions that the linear model can express. 
    }
    \label{fig:rkhs}
\end{figure*}

Once one gets used to this definition, it is immediately apparent that the functions living in the RKHS of the quantum kernel are what we defined as quantum models:
\begin{theorem}
Functions in the RKHS $F$ of the quantum kernel are linear models in the data-encoding feature space $\mathcal{F}$ and vice versa.
\end{theorem}
\begin{proof}
The functions in the RKHS of the quantum kernel are of the form $f(\cdot) = \sum_k \gamma_k \kappa(x^k, \cdot)$, with $x^k \in \mathcal{X}$. We get
\begin{align}
    f(x) &= \sum_k \gamma_k \kappa(x^k, x) \\
    &= \sum_k \gamma_k \tr{\rho(x^k) \rho(x) }\\
    &= \tr{\sum_k \gamma_k \rho(x^k) \rho(x) }\\
    &= \tr{\meas \rho(x) }.
\end{align}
Using Theorem~\ref{thm:measurement_expansion_general} we know that all quantum models can be expressed by measurements $\sum_k \gamma_k \rho(x^k)$, and hence by functions in the RKHS.
\end{proof}
In fact, the above observation applies to \textit{any} linear model in a feature space that gives rise to the quantum kernel (see  Theorem 4.21 in \citep{steinwart2008support}).

% \begin{theorem}[Construction of an RKHS from a linear model in feature space] \label{thm:rkhs_models}
% Let $\mathcal{X} \neq \emptyset$ and $\kappa$ be a kernel over $\mathcal{X}$ with feature space $\mathcal{F}$ and feature map $\phi_{0}$. Then 
% \begin{equation}H := \{ f: \mathcal{X} \rightarrow \mathbbm{R} | \exists \;   w \in \mathcal{F} \text{ with } f(x) = \langle w |  \phi(x) \rangle_{\mathcal{F}} \text{ for all }  x \in \mathcal{X} \},\end{equation}
% equipped with the norm
% \begin{equation}\Vert f \Vert_{\mathcal{F}} = \inf \{  \Vert w \Vert_{\mathcal{F}}  | \;  w \in \mathcal{F} \text{ with } f(\cdot) = \langle w,  \phi_0(\cdot ) \rangle_{\mathcal{F}} \}.\end{equation}
% is the (only) RKHS for which $\kappa$ is a reproducing kernel. 
% \end{theorem}

As a first taste of how the connection of quantum models and kernel theory can be exploited for quantum machine learning, consider the question whether quantum models are universal function approximators. If quantum models are universal, the RKHS of the quantum kernel must be universal (or dense in the space of functions we are interested in). This leads to the definition of a universal kernel (see \cite{steinwart2008support} Definition 4.52):
\begin{definition}[Universal kernel]
A continuous kernel $\kappa$ on a compact metric space $\mathcal{X}$ is called
universal if the RKHS $F$ of $\kappa$ is dense in $C(\mathcal{X})$, i.e., for every function
$g$ in the set of functions $C(\mathcal{X})$ mapping from elements in $\mathcal{X}$ to a scalar value, and for all $\epsilon > 0$ there exists an $f \in F$ such that
\begin{equation}  \Vert f - g  \Vert_{\infty} \leq \epsilon.\end{equation}
\end{definition} 

The reason why this is useful is that there are a handful of known necessary conditions for a kernel to be universal, for example if its feature map is injective (see \cite{steinwart2008support} for more details). This immediately excludes quantum models defined on the data domain $\mathcal{X} = \mathbb{R}$ which use single-qubit Pauli rotation gates of the form $e^{ix \sigma}$ (with $\sigma$ a Pauli matrix) to encode data: since such rotations are $2\pi$-periodic, two different $x, x' \in \mathcal{X}$ get mapped to the same data-encoding state $\rho(x)$. In other words, and to some extent trivially so, on a data domain that extends beyond the periodicity of a quantum model we never have a chance for universal function approximation. Another example for universal kernels are kernels of the form $\kappa(x, x') = \sum_{k=1}^{\infty} c_k \langle x', x\rangle^k$ (see \cite{steinwart2008support} Corollary 4.57). Vice versa, the universality proof for a type of quantum model in \cite{schuld2020effect} suggests that some quantum kernels of the form~(\ref{thm:kernel_fourier}) are universal in the asymptotic limit of exponentially large circuits.

I want to finish with a final note about the relation between ``wavefunctions'' and functions in the RKHS of quantum systems (see also the appendix of \cite{schuld2019quantum}). Quantum states are sometimes called ``wavefunctions'', since an alternative definition of the Hilbert space of a quantum system is the space of functions $f(\cdot) = \psi(\cdot)$ which map a measurement outcome $i$ corresponding to basis state $\ket{i}$ to an ``amplitude'' $\psi(i) = \braket{i}{\psi}$. (The dual basis vector $\bra{i}$ can here be understood as the evaluating functional $\delta_{i}$ which returns this amplitude.) Hence, the Hilbert space of a quantum system can be written as a space of functions mapping from $\{i\} \rightarrow \mathbbm{C}$. But the functions that we are interested in for machine learning are functions \textit{in the data}, not in the possible measurement outcomes. This means that the Hilbert space of the quantum system is only equivalent to the RKHS of a quantum machine learning model if we associate data with the measurement outcomes. This is true for many proposals of generative quantum machine learning models \cite{benedetti2019generative, cheng2018information}, and it would be interesting to transfer the results to this setting.

\section{Training quantum models}\label{sec:optimisation}

While the question of universality addresses the expressivity of quantum models, the remaining sections will look at questions of trainability and optimisation, for which the kernel perspective has the most important results to offer. Notably, we will see that the optimal measurements of quantum models for typical machine learning cost functions only have relatively few degrees of freedom. Similarly, the process of finding these optimal models (i.e., training over the space of all possible quantum models) can be formulated as a low-dimensional optimisation problem. Most of the results are based on the fact that for kernel methods, the task of training a model is equivalent to optimising over the model's corresponding RKHS.

\subsection{Optimising quantum models is equivalent to optimising over the RKHS}\label{ssec:optimsing_rkhs}

In machine learning we want to find optimal models, or those that minimise the cost functions derived from learning problems. This process is called \textit{training}. From a learning theory perspective, training can be phrased as \textit{regularised empirical risk minimisation}, and the problem of training quantum models can be cast as follows:

\begin{definition}[Regularised empirical risk minimisation of quantum models]
Let $\mathcal{X}, \mathcal{Y}$ be data input and output domains, $p$ a probability distribution on $\mathcal{X}$ from which data is drawn, and $L: \mathcal{X} \times \mathcal{Y} \times \mathbb{R} \rightarrow [0, \infty) $ a loss function that quantifies the quality of the prediction of a quantum model $f(x)= \tr{\rho(x) \meas}$. Let 
\begin{equation}
\mathcal{R}_L(f) = \int_{\mathcal{X} \times \mathcal{Y}} L(x, y, f(x)) \,\mathrm{d}p(x, y) 
\end{equation}
be the expected loss (or ``\textit{risk}'') of $f$ under $L$, where $L$ may depend explicitly on $x$. Since $p$ is unknown, we approximate the risk by the empirical risk  
\begin{equation}\hat{\mathcal{R}}_L(f) = \frac{1}{M} \sum_{m=1}^M L(x^m, y, f(x^m)).\end{equation}
Regularised empirical risk minimisation of quantum models is the problem of minimising the empirical risk over all possible quantum models while also minimising the norm of the measurement $\meas$, 
\begin{equation}\label{eq:erm}
    \inf_{\meas \in \mathcal{F}} \lambda \| \meas \|^2_{\mathcal{F}} + \hat{\mathcal{R}}_L(\tr{\rho(x) \meas}),
\end{equation}
where $\lambda \in \mathbb{R}^+$ is a positive hyperparameter that controls the strength of the regularisation term.
\end{definition}

We saw in Section~\ref{sec:models_rkhs} that quantum models are equivalent to functions in the RKHS of the quantum kernel, which allows us to replace the term $\hat{\mathcal{R}}_L(\tr{ \rho(x) \meas})$ in the empirical risk by $\hat{\mathcal{R}}_L(f)$, $f \in F$.

But what about the regularisation term? Since with Theorem~\ref{thm:measurement_expansion_general} we can write
\begin{align}
 \| \meas \|^2_{\mathcal{F}}  &= \tr{\meas^2} \\
 &= \sum_{ij} \gamma_i \gamma_j \tr{\rho(x^i)\rho(x^j)}\\
&= \sum_{ij} \gamma_i \gamma_j \kappa(x^i, x^j) \\
&=  \langle \sum_{i} \gamma_i \kappa(x^i, \cdot), \sum_{i} \gamma_i \kappa(x^i, \cdot) \rangle_{F}\\
&=\langle f, f \rangle_{F},
\end{align}
the norm of $\meas \in \mathcal{F}$ is equivalent to the norm of a corresponding $f \in {F}$. Hence, the regularised empirical risk minimisation problem in Eq.~(\ref{eq:erm}) is equivalent to  
\begin{equation}\label{eq:erm2}
    \inf_{f \in F } \gamma  \Vert f \Vert^2_{F} + \hat{\mathcal{R}}_L(f),
\end{equation}
which minimises the regularised risk over the RKHS of the quantum kernel. We will see in the remaining sections that this allows us to characterise the problem of training and its solutions to a surprising degree.

\subsection{The measurements of optimal quantum models are expansions in the training data}\label{ssec:meas_expansion}

The \textit{representer theorem}, one of the main achievements of classical kernel theory, prescribes that the function $f$ from the RKHS which minimises the regularised empirical risk can always be expressed as a weighted sum of the kernel between $x$ and the training data. Together with the connection between quantum models and the RKHS of the quantum kernel, this fact will allow us to write optimal quantum machine learning models in terms of the quantum kernel.

More precisely, the representer theorem can be stated as follows (for a more general version, see \cite{scholkopf2002learning}, Theorem 5.1):
\begin{theorem}[Representer theorem]\label{thm:representer_quantum}
Let $\mathcal{X}, \mathcal{Y}$ be an input and output domain, $\kappa : \mathcal{X} \times \mathcal{X} \to \mathbb{R}$ a kernel with a corresponding reproducing kernel Hilbert space $F$, and given training data $\mathcal{D} = \{(x^1, y^1), \dotsc, (x^M, y^M) \in \mathcal{X} \times \mathcal{Y}\}$. Consider a strictly monotonic increasing regularisation function $g \colon [0, \infty) \to \mathbb{R}$, and an arbitrary loss $L \colon \mathcal{X} \times \mathcal{Y}\times \mathbb{R} \to \mathbb{R}\cup \lbrace \infty \rbrace$. Any minimiser of the regularised empirical risk
\begin{equation}
 f_{\rm opt} = \underset{f \in F}{\mathrm{argmin}} \left\lbrace \hat{\mathcal{R}}_L(f) + g\left( \lVert f \rVert_F \right) \right \rbrace, \quad 
\end{equation}
admits a representation of the form:
\begin{equation}
  f_{\rm opt}(x) = \sum_{m = 1}^M \alpha_m \; \kappa(x^m, x),
\end{equation}
where $\alpha_m \in \mathbb{R}$ for all $1 \le m \le M$.
\end{theorem}
Note that the crucial difference to the form in Theorem~(\ref{thm:measurement_expansion_general}) is that $m$ does not sum over arbitrary data from $\mathcal{X}$, but over a finite training data set.
For us this means that the optimal quantum model can be written as
\begin{equation}\label{eq:optimal_model}
    f_{\rm opt}(x) =\sum_{m = 1}^M \alpha_m \; \tr{\rho(x) \rho(x^m)} =  \sum_{m = 1}^M \alpha_m \; |\langle \phi(x) | \phi(x^m) \rangle |^2.
\end{equation}
This in turn defines the measurements $\meas$ of optimal quantum models.
\begin{theorem}[Optimal measurements] \label{thm:opt_meas}
For the settings described in Theorem~\ref{thm:representer_quantum}, the measurement that minimises the regularised empirical risk can be written as an expansion in the training data $x^m$, $m=1 \dots M$,
\begin{equation}\label{eq:opt_meas} 
 \meas_{\rm opt} =   \sum_{m}  \alpha_{m}  \rho(x^{m}),
\end{equation}
with $\alpha_m \in \mathbb{R}$.
\end{theorem}
\begin{proof}
This follows directly by noting that
\begin{align}
     f_{\rm opt}(x) & =  \sum_{m = 1}^M \alpha_m \; \tr{ \rho(x)  \rho(x^m)}\\
     & =  \tr{ \rho(x) \sum_{m = 1}^M \alpha_m \rho(x^m)}\\
     & = \tr{ \rho(x)\meas_{\rm opt}}
\end{align}
\end{proof}
As mentioned in the summary and Figure~\ref{fig:training}, in variational circuits we typically only optimise over a subspace of the RKHS since the measurements $\meas$ are constrained by a particular circuit ansatz. We can therefore not guarantee that the optimal measurement can be expressed by the variational ansatz. However, the above guarantees that there will always be a measurement of the form of Eq.~(\ref{eq:opt_meas}) for which the quantum model has a lower regularised empirical risk than the best solution of the variational training.

As an example, we can use the apparatus of linear regression to show that the optimal measurement for a quantum model under least-squares loss can indeed be written as claimed in Eq.~(\ref{eq:opt_meas}). For this I will assume once more that $\mathcal{X} = \mathbb{R}^{N}$ where $N = 2^n$ and $n$ is the number of qubits, and switch to bold notation. I will also use the (here much more intuitive) vectorised notation in which the quantum model $f(x) = \tr{\rho(x) \meas}$ becomes $f(x) = \langle \braakett{ \meas }{ \rho(x) }$, with the vectorised measurement $\kett{\meas } = \sum_k \gamma_k  \kett{ \rho(x^k) } $. 

A well-known result from linear regression states that the vector $\mathbf{w}$ that minimises the least-squares loss of a linear model $f(\mathbf{x}) = \mathbf{w}^T \mathbf{x}$ is given by
\begin{equation}
    \mathbf{w} = (\mathbf{X}^{\dagger}\mathbf{X})^{-1} \mathbf{X}^{\dagger} \mathbf{y}, 
\end{equation} 
if the inverse of $\mathbf{X}^{\dagger}\mathbf{X}$ exist. Here, $\mathbf{X}$ is the matrix that contains the data vectors as rows,
\begin{equation}\mathbf{X} = 
\begin{pmatrix}  
x^1_1 & \hdots & x^1_N \\ 
\vdots & \ddots & \vdots\\
x^M_1 & \hdots & x^M_N 
\end{pmatrix},\end{equation}
and $\mathbf{y}$ is an $M$-dimensional vector containing the target labels. A little trick exposes that $\mathbf{w}$ can be written as a linear combination of training inputs,
\begin{equation} \label{eq:trick_linear_expansion}
\mathbf{w} = \mathbf{X}^{\dagger}\left( \mathbf{X} (\mathbf{X}^{\dagger}\mathbf{X})^{-2} \mathbf{X}^{\dagger} \mathbf{y}\right) = \mathbf{X}^{\dagger} \boldsymbol{\alpha} = \sum_m \alpha_m \mathbf{x}^m,
\end{equation}
where $\boldsymbol{\alpha} = (\alpha_1,\dots, \alpha_M)$.

Since a quantum model is a linear model in feature space, we can associate the vectors in linear regression with the vectorised measurement and density matrix, and immediately derive 
\begin{equation}
    \kett{\meas } = \sum_m  y^m\left( \sum_{m'} \kettbraa{\rho(\mathbf{x}^{m'})}{\rho(\mathbf{x}^{m'}) }  \right)^{-1}  \kett{ \rho(\mathbf{x}^m) } ,
\end{equation}
by making use of the fact that in our notation
\begin{equation}\mathbf{X}^{\dagger}\mathbf{X} \Longleftrightarrow \sum_m  \kettbraa{\rho(\mathbf{x}^{m})}{\rho(\mathbf{x}^{m}) } ,  \end{equation}
and 
\begin{equation}\mathbf{X}^{\dagger}\mathbf{y} \Longleftrightarrow \sum_m y^m \kett{\rho(\mathbf{x}^m)}.
\end{equation}
Note that although this looks like an expansion in the feature states, the ``coefficient'' of $ \kett{\rho(\mathbf{x}^m)}$ still contains an operator. However, with Eq.~(\ref{eq:trick_linear_expansion})
and writing $\sum_m \kettbraa{\rho(\mathbf{x}^m)}{\rho(\mathbf{x}^m) }$ in its diagonal form,
\begin{equation}\sum_m \kettbraa{\rho(\mathbf{x}^m) }{ \rho(\mathbf{x}^m) } = \sum_k h_k \kettbraa{ h_k}{ h_k }, \end{equation}
we have
\begin{equation}
    \kett{ \meas } = \sum_m  \alpha_m \kett{\rho(\mathbf{x}^m) },
\end{equation}
with 
\begin{equation}\alpha_m  = \sum_k h^{-2}_k {\braakett{h_k}{\rho(\mathbf{x}^{m})}} \sum_{m'} y^{m'} \braakett{h_k}{\rho(\mathbf{x}^{m'})}.\end{equation}
The optimal measurement in ``matrix form'' reads
\begin{equation}
    \meas =  \sum_m \alpha_m \rho(\mathbf{x}^m)= \sum_m \alpha_m \ketbra{\phi(\mathbf{x}^m)}{\phi(\mathbf{x}^m)},
\end{equation}
as claimed by the representer theorem.

Of course, it may require a large routine to implement this measurement fully quantumly, since it involves inverting operators acting on the feature space. Alternatively one can compute the desired $\{\alpha_m\}$ classically and use the quantum computer to just measure the kernel. In the last section we will see ideas of how to use quantum algorithms to do the inversion, but these quantum training algorithms are complex enough to require fault-tolerant quantum computers which we do not have available today.

\subsection{The kernel defines which models are punished by regularisation}\label{ssec:regularisation}

In statistical learning theory, the role of the regulariser in the regularised empirical risk minimisation problem is to ``punish'' some functions and favour others. Above, we specifically looked at regularisers of the form $\|f\|^2_{F}$, $f \in F$, which was shown to be equivalent to minimising the norm of the measurement (or the length of the vectorised measurement) in feature space. But what is it exactly that we are penalising here? It turns out that the kernel does not only fix the space of quantum models themselves, it also defines which functions are penalised in regularised empirical risk minimisation problems. This is beautifully described in \cite{scholkopf2002learning} Section 4.3, and I will only give a quick overview here. 

To understand regularisation, we need to have a closer look at the regularising term $\|f\|^2_{F} = \langle f, f \rangle_{F} $. But with the construction of the RKHS it actually remains very opaque what this inner product actually computes. It turns out that for every RKHS $F$ there is a transformation $\Upsilon: F \rightarrow L_2(\mathcal{X})$ that maps functions in the RKHS to square integrable functions on $\mathcal{X}$. What we gain is a more intuitive inner product formed by an integral,
\begin{equation}
    \langle f, f \rangle_{F} = \langle \Upsilon f, \Upsilon f \rangle_{L_2} = \int_{\mathcal{X}} (\Upsilon f(x))^2 dx. 
\end{equation}
The operator $\Upsilon$ can be understood as extracting the information from the model $f$ which gets integrated over in the usual $L_2$ norm, and hence penalised during optimisation. For example, for some kernels this can be shown to be the derivative of functions, and regularisation therefore provably penalise models with ``large'' higher-order derivatives -- which means it favours smooth functions. 

The important point is that every kernel defines a unique transformation $\Upsilon$, and therefore a unique kind of regularisation. This is summarised in Theorem 4.9 in \cite{scholkopf2002learning}, which I will reprint here without proof:
\begin{theorem}[RKHS and Regularization Operators] For every RKHS with reproducing kernel $\kappa$ there exists a corresponding regularization operator $\Upsilon : F \to D$ (where $D$ is an inner product space) such that for all $f \in F$, 
 \begin{equation}
     \langle \Upsilon \kappa(x, \cdot), \Upsilon f(\cdot) \rangle_{D} = f(x),
 \end{equation}
 and in particular
 \begin{equation}\langle \Upsilon \kappa(x, \cdot), \Upsilon \kappa(x', \cdot) \rangle_{D} = \kappa(x, x').
 \end{equation}
    Likewise, for every regularization operator $\Upsilon : F\to D$, where $F$ is some function space equipped with a dot product, there exists a corresponding RKHS $F$ with reproducing kernel $\kappa$ such that these two equations are satisfied.
\end{theorem}
In short, the quantum kernel or data-encoding strategy does not only tell us about universality and optimal measurements, it also fixes the regularisation properties in empirical risk minimisation. Which data encoding actually leads to which regularisation property is still an interesting open question for research.

\subsection{Picking the best quantum model is a low-dimensional (convex) optimisation problem}\label{ssec:convex}

Besides the representer theorem, a second main achievement of kernel theory is to recognise that optimising the empirical risk of convex loss functions over functions in an RKHS can be formulated as a finite-dimensional convex optimisation problem (or in less cryptic language, optimising over extremely large spaces is surprisingly easy when we use training data, something noted in \cite{huang2020power} before). 

The fact that the optimisation problem is finite-dimensional -- and we will see the dimension is equal to the number of training data -- is important, since the feature spaces in which the model classifies the data are usually very high-dimensional, and possibly even infinite-dimensional. This is obviously true for the data-encoding feature space of quantum computations as well -- which is precisely why variational quantum machine learning parametrise circuits with a small number of trainable parameters instead of optimising over all unitaries/measurements. But even if we optimise over all quantum models, the results of this section guarantee that the dimensionality of the problem is limited by the size of the training data set. 

The fact that optimisation is convex means that there is only one global minimum, and that we have a lot of tools to find it \citep{boyd2004convex} - in particular, more tools than mere gradient descent. Convex optimisation problems can be roughly solved in time $\mathcal{O}(M^2)$ in the number of training data. Although prohibitive for large datasets, it makes the optimisation guaranteed to be tractable (and below we will see that quantum computers could in principle help to train with a runtime of $\mathcal{O}(M)$). 

Let me make the statement more precise. Again, it follows from the fact that optimising over the RKHS of the quantum kernel is equivalent to optimising over the space of quantum models.  

\begin{theorem}[Training quantum models can be formulated as a finite-dimensional convex program]\label{thm:convex}
Let $\mathcal{X}$ be a data domain and $\mathcal{Y}$ an output domain, $L: \mathcal{X} \times \mathcal{Y} \times \mathbb{R} \rightarrow [0, \infty)$ be a loss function, $F$ the RKHS of the quantum kernel over a non-empty convex set $\mathcal{X}$ with the reproducing kernel $\kappa$. Furthermore, let $\lambda \geq 0$ be a regularisation parameter and $D = \{(x^m, y^m), m=1,\dots,M\} \subset \mathcal{X} \times \mathcal{Y}$ a training data set. The regularised empirical risk minimisation problem is finite-dimensional, and if the loss is convex, it is also convex. 
\end{theorem}

\begin{proof}
Recall that according to the Representer Theorem~\ref{thm:representer_quantum}, the solution to the regularised empirical risk minimisation problem 
\begin{equation}f_{\rm opt} = \inf_{f \in F} \lambda \Vert f \Vert^2_{F} + \hat{\mathcal{R}}_{L}(f)  \end{equation} 
has a representation of the form 
\begin{equation}f_{\rm opt}(x) = \sum_m \alpha_m \tr{\rho(x^m) \rho(x)}.\end{equation}
We can therefore write 
\begin{equation}
    \hat{\mathcal{R}}_{L}(f) = \frac{1}{M} \sum_{m} L(x^m, y^m, \sum_{m'} \alpha_{m'} \kappa(x^m, x^{m'})).
\end{equation}
If the loss $L$ is convex, then this term is also convex, and it is $M$-dimensional since it only involves the $M$ degrees of freedom $\alpha_m$. 

Now let us turn to the regularisation term and try to show the same. Consider
\begin{equation} \Vert f \Vert^2_{F} = \sum_{m, m'} \alpha_m \alpha_{m'} \tr{\rho(x^m)\rho(x^{m'})} =  \sum_{m, m'} \alpha_m \alpha_{m'} \kappa(x^m, x^{m'}) = \boldsymbol\alpha^T \mathbf{K} \boldsymbol\alpha, \end{equation}
where $\mathbf{K} \in \mathbb{R}^{M \times M}$ is the kernel matrix or \textit{Gram matrix} with entries $K_{m, m'} = \kappa(x^m, x^{m'})$, and $\boldsymbol\alpha = (\alpha_1, \dots, \alpha_M)$ is the vector of coefficients $\alpha_m$. Since $\mathbf{K}$ is by definition of the kernel positive definite, this term is also convex. Both $\boldsymbol\alpha$ and $\mathbf{K}$ are furthermore finite-dimensional.

Together, training a quantum model to find the optimal solution from Eq.~(\ref{eq:optimal_model}) can be done by solving the optimisation problem
\begin{align} \label{eq:opt_problem}
    \inf_{\boldsymbol\alpha \in \mathbb{R}^M} \frac{1}{M} \sum_{m} L(x^m, y^m, \sum_{m'} \alpha_{m'} \kappa(x^m, x^{m'}))  + \lambda \boldsymbol\alpha^T \mathbf{K} \boldsymbol\alpha,
\end{align}
which optimises over $M$ trainable parameters, and is convex for convex loss functions.
\end{proof}

A \textit{support vector machine} is a special case of kernel-based training which uses a special convex loss function, namely the hinge loss, for $L$: 
\begin{equation}
    L(f(x), y) = \max(0, 1- f(x)y),
\end{equation}
where one assumes that $y \in \{-1, 1\}$. As derived in countless textbooks, the resulting optimisation problem can be constructed from geometric arguments as maximising the ``soft'' margin of the closest vectors to a decision boundary. Under this loss, Eq.~(\ref{eq:opt_problem}) reduces to
\begin{equation}\label{eq:alphas}
    \mathbf{\alpha}_{\rm opt} = \max_{\mathbf{\alpha}}  \; \sum_m \alpha_m - \frac{1}{2} \sum_{m, m'} \alpha_{m} \alpha_{m'} y^m y^{m'} \kappa(x^{m}, x^{m'}).
\end{equation}
Training a support vector machine with hinge loss and a quantum kernel $\kappa$ is equivalent to finding the general quantum model that minimises the hinge loss. The ``quantum support vector machine'' in \citep{schuld2019quantum, havlivcek2019supervised} is therefore not one of many ideas to build a hybrid classifier, it is a generic blueprint of how to train quantum models in a kernel-based manner. 

\section{Should we switch to kernel-based quantum machine learning?}

The fact that quantum models can be formulated as kernel methods with a quantum kernel raises an important question for current quantum machine learning research: how do kernel-based models, i.e., solutions to the problem in Eq.~(\ref{eq:opt_problem}), compare to models whose measurements are trained variationally? Let us revisit Figure~\ref{fig:training} in light of the results of the previous section. 

We saw in Section \ref{ssec:convex} how kernel-based training optimises the measurement over a subspace spanned by $M$ encoded training inputs by finding the best coefficients $\alpha_m$, $m=1\dots M$. We also saw in Section \ref{ssec:meas_expansion} that this subspace contains the globally optimal measurement. Variational training instead optimises over a subspace defined by the parametrised ansatz, which may or may not overlap with the training-data subspace, and could therefore not have access to the global optimum. The advantages of kernel-based training are therefore that we are guaranteed to find the globally optimal measurement over all possible quantum models. If the loss is convex, the optimisation problem is furthermore of a favourable structure that comes with a lot of guarantees about the performance and convergence of optimisation algorithms. But besides these great properties, in classical machine learning with big data, kernel methods were superseded by neural networks or approximate kernel methods \citep{rahimi2007random} because of their poor scaling. Training involves computing the pair-wise distances between all training data in the Gram matrix of Eq.~(\ref{eq:opt_problem}), which has at least a runtime of $\mathcal{O}(M^2)$ in the number of training samples $M$.\footnote{Note that this is also true when using the trained model for predictions, where we need to compute the distance between a new input to any training input in feature space as shown in Eq.~(\ref{eq:optimal_model}). However, in maximum margin classifiers, or support vector machines in the stricter sense, most $\alpha_m$ coefficients are zero, and only the distances to a few ``support vectors'' are needed.} In contrast, training neural networks takes time $\mathcal{O}(M)$ that only depends linearly on the number of training samples. Can the training of variational quantum circuits offer a similar advantage over kernel-based training?

The answer is that it depends. So far, training variational circuits with gradient-based methods on hardware is based on so-called parameter-shift rules \cite{mitarai2018quantum, schuld2019evaluating} instead of backpropagation. This strategy introduces a linear scaling with the number of parameters $|\theta|$, and the number of circuits that need to be evaluated to train a variational quantum model therefore grows with $\mathcal{O}(|\theta| M)$. If the number of parameters in an application grows sufficiently slowly with the dataset size, variational circuits will almost be able to match the good scaling behaviour of neural networks, which is an important advantage over kernel-based training. But if, like in neural networks, the number of parameters in a variational ansatz grows linearly with the number of data, variational quantum models end up having the same quadratic scaling as the kernel-based approach regarding the number of circuits to evaluate. Practical experiments with $10-20$ parameters and about $100$ data samples show that the constant overhead of gradient calculations on hardware make kernel-based training in fact much faster for small-scale applications.\footnote{See \url{https://pennylane.ai/qml/demos/tutorial_kernel_based_training.html}.} In addition, there is  no guarantee that the final measurement is optimal, we have high-dimensional non-convex training landscapes, and the additional burden of choosing a good variational ansatz. In conclusion, the kernel perspective is not only a powerful and theoretically appealing alternative to think about quantum machine learning, but may also speed up current quantum machine learning methods significantly. 

As a beautiful example of the mutually beneficial relation of quantum computing and kernel methods, the story does not end here. While all of the above is based on models evaluated on a quantum computer but trained classically, convex optimisation problems happen to be exactly the kind of thing quantum computers are good at \citep{harrow2009quantum}. We can therefore ask whether quantum models could not in principle be \textit{trained} by quantum algorithms. ``In principle'' alludes to the fact that such algorithms would likely be well beyond the reach of near-term devices, since training is a more complex affair that requires fully error-corrected quantum computers which we do not have yet.

The reasons why quantum training could help to lower this scaling are hidden in results from the early days of quantum machine learning, when quantum-based training was actively studied in the hope of finding exponential speedups for classical machine learning \citep{wiebe2012quantum, rebentrost2014quantum, lloyd2014quantum}. While these speedups only hold up under very strict assumptions of data loading oracles, they imply quadratic speedups for rather general settings (see also Appendix~\ref{app:outlook}). They can be summarised as follows: \textit{given a feature map implemented by a fault-tolerant quantum computer, we can train kernel methods in time that grows linearly in the data. } If a kernel can be implemented as a quantum computation (like the Gaussian kernel \citep{chatterjee2016generalized}), this speedup would also hold for ``classical models'' -- which are then merely run on a quantum computer. 
 
Of course, fault-tolerant quantum computers may still take many years to develop and are likely to have a large constant overhead due to the expensive nature of quantum error correction. But in the longer term, this shows that the use of quantum computing is not only to implement interesting kernels. Quantum computers have the potential to become a game changer for kernel-based machine learning in a similar way to how GPU-accelerated hardware enabled deep learning.

\section*{Acknowledgements}
I want to thank Johannes Jakob Meyer, Nathan Killoran, Olivia di Matteo and Filippo Miatto, Nicolas Quesada and Ilya Sinayskiy for their time and helpful comments.

%\bibliography{lit.bib}


\appendix

\section{Proof of Theorem~\ref{thm:kernel_fourier}} \label{app:fourier_proof}

First, note that we are able to assume without loss of generality that the encoding generator $G$ is diagonal because one can diagonalise Hermitian operators as $G = V e^{-i x_i \Sigma} V^{\dagger}$ with \begin{equation}
    e^{-i x_i \Sigma} = 
    \begin{pmatrix}
    e^{-i x_i \lambda_1} & \vdots & 0 \\
    0 & \ddots &  \\
    0 & \vdots &  e^{-i x_i \lambda_d}
    \end{pmatrix}
\end{equation}
where $\{\lambda_1, \dots, \lambda_d\}$ are the eigenvalues of $G$. Formally one can ``absorb'' $V, V^{\dagger}$ into the arbitrary circuits $W$ before and after the encoding gate. The remainder is just a matter of writing the matrix multiplications that represent the quantum circuit as a sum in the computational basis, and trying to introduce notation that hides irrelevant complexity:

\begin{align} \label{eq:eff_kernel_fourier}
    \kappa (\mathbf{x}, \mathbf{x}') &=  |\braket{\phi(\mathbf{x}')}{\phi(\mathbf{x})} |^2 \\
    &= \left| \bra{0} (W^{(1)})^{\dagger} (e^{-i x'_1 \Sigma} )^{\dagger} \cdots (e^{-i x'_N \Sigma} )^{\dagger} \underbrace{(W^{(N+1)})^{\dagger}  W^{(N+1)}}_{\mathbbm{1}} e^{-i x_N \Sigma}  \cdots e^{-i x_1 \Sigma}  W^{(1)}\ket{0}\right|^2\\
    &= \left|\bra{0} (W^{(1)})^{\dagger} (e^{-i x'_1 \Sigma} )^{\dagger} \cdots (e^{-i x'_N \Sigma} )^{\dagger}  e^{-i x_N \Sigma}  \cdots e^{-i x_1 \Sigma}  W^{(1)}\ket{0}\right|^2\\
    &= \left|\sum_{j_1, \dots, j_N = 1}^{d} \sum_{k_1, \dots, k_N = 1}^{d}  e^{-i \left( \lambda_{j_1}x_1 - \lambda_{k_1}x'_1 + \dots + \lambda_{j_N}x_N - \lambda_{k_N}x'_N) \right)} \left( W^{(1)}_{1 k_1} \dots W^{(N)}_{k_{N-1} k_N} \right)^* W^{(N)}_{j_N j_{N-1}}   \dots W^{(1)}_{j_1 1} \right|^2\\
     &= \left|\sum_{\mathbf{j}} \sum_{\mathbf{k}} e^{-i \left( \Lambda_{\mathbf{j}}\mathbf{x} - \Lambda_{\mathbf{k}}\mathbf{x}' \right) } (w_{\mathbf{k}})^* w_{\mathbf{j}} \right|^2\\
     &= \sum_{\mathbf{j}} \sum_{\mathbf{k}}\sum_{\mathbf{h}} \sum_{\mathbf{l}}  e^{-i \left( \Lambda_{\mathbf{j}} - \Lambda_{\mathbf{l}} \right) \mathbf{x} } e^{i \left( \Lambda_{\mathbf{k}} - \Lambda_{\mathbf{h}} \right) \mathbf{x}' } (w_{\mathbf{k}}w_{\mathbf{h}})^* w_{\mathbf{j}}w_{\mathbf{l}} 
\end{align}
Here, the scalars $W^{(i)}_{a b}$, $i=1, \dots, N$, refer to the element $\bra{a} W^{(i)}\ket{b}$ of the unitary operator $W^{(i)}$, the bold multi-index $\mathbf{j}$ summarises the set $(j_1, \dots, j_N)$ where $j_i \in \{1, \dots, d\}$ and $\Lambda_{\mathbf{j}}$ is a vector containing the eigenvalues selected by the multi-index (and similarly for $\mathbf{k}, \mathbf{h}, \mathbf{l}$).\\

We can now summarise all terms where $ \Lambda_{\mathbf{j}} - \Lambda_{\mathbf{l}} = \mathbf{s}$ and $ \Lambda_{\mathbf{k}} - \Lambda_{\mathbf{h}} = \mathbf{t}$, in other words where the differences of eigenvalues amount to the same vectors $\mathbf{s}, \mathbf{t}$. Then
\begin{align}
      \kappa (\mathbf{x}, \mathbf{x}')   &= \sum_{\mathbf{s}, \mathbf{t} \in \Omega}  e^{-i \mathbf{s} \mathbf{x} } e^{i \mathbf{t} \mathbf{x}' } \sum_{\mathbf{j}, \mathbf{l} | \Lambda_{\mathbf{j}}-\Lambda_{\mathbf{l}} = \mathbf{s} } \; 
      \sum_{\mathbf{k}, \mathbf{h} | \Lambda_{\mathbf{k}}-\Lambda_{\mathbf{h}} = \mathbf{t} } w_{\mathbf{j}}w_{\mathbf{l}} (w_{\mathbf{k}}w_{\mathbf{h}})^*\\
      &= \sum_{\mathbf{s}, \mathbf{t} \in \Omega}  e^{-i \mathbf{s} \mathbf{x} } e^{i \mathbf{t} \mathbf{x}' }c_{\mathbf{st}}.
\end{align}
The frequency set $\Omega$ contains all vectors $\{ \Lambda_{\mathbf{j}} - \Lambda_{\mathbf{k}}\}$ with $\Lambda_{\mathbf{j}} = (\lambda_{j_1}, \dots , \lambda_{j_N}) $, $j_1,\dots, j_N \in [1, \dots, d] $.
Let me illustrate this rather unwieldy notation with our standard example of encoding a real scalar input $x$ via a Pauli-X rotation. 
\begin{example}
Consider the embedding from Example~\ref{ex:enc_kernel}. We have $W^{(1)} = W^{(2)} = \mathbbm{1}$. With a singular value decomposition one can write the rotation operator as 
\begin{equation}
R_x(x) = e^{-i x\frac{1}{2} \sigma_x} = V^{\dagger} e^{-i x\frac{1}{2} \Sigma} V,
\end{equation}
with 
\begin{equation}
V =\frac{1}{\sqrt{2}} \begin{pmatrix} 1& 1 \\ -1& 1\end{pmatrix}.
\end{equation}
The unitary operators $V, V^{\dagger}$ can be absorbed into the general unitaries applied before and after the encoding, which sets $W^{(1)} = V^{\dagger}$ and $W^{(2)} = V$. The remaining $\frac{1}{2} \Sigma$ is a diagonal operator with eigenvalues $\{\lambda_1=-\frac{1}{2}, \lambda_2=\frac{1}{2}\}$. We get
\begin{equation}
    \kappa(x, x') 
    = \left|\sum_{j = 1}^{2} \sum_{k= 1}^{2} \sum_{i=1}^2 e^{-i (\lambda_{j}x- \lambda_{k}x') } \; (V_{1 k})^*  (V_{k i})^*  V_{i j} V_{j 1} \right|^2.
\end{equation}
Due to unitarity, inner products of different rows/columns of $V$, $V^{\dagger}$ are zero, and so  $\sum_{i=1}^2 (V_{k i})^*  V_{i j} = \delta_{k j}$, leading to
\begin{align} 
    \kappa(x, x') 
    &= \left|\sum_{j = 1}^{2}  e^{-i \lambda_{j}(x- x') } \; (V_{1 j})^* V_{j 1} \right|^2\\
    &= \left|e^{-i \lambda_1(x-x' )} \; (V_{1 1})^* V_{1 1}  
        + e^{-i \lambda_2(x-x') } \; (V_{1 2})^* V_{2 1}\right|^2  \\ 
    &= \left|\frac{1}{2}e^{i \frac{x- x'}{2}}  + \frac{1}{2}e^{-i \frac{x - x'}{2} }  \right|^2\\
    &=| \cos\left(\frac{x- x'}{2}\right)|^2\\
    &= \cos^2\left(\frac{x- x'}{2}\right).
\end{align}
This is the same result as in the ``straight'' computation from Eq.~(\ref{eq:example_model}).
\end{example}

\section{Convex optimisation with quantum computers}\label{app:outlook}

The family of quantum algorithms for convex optimisation in machine learning consists of many variations, but is altogether based on results that establish fast linear algebra processing routines for quantum computers. They are very technical in design, which is why they may not be easily accessible to many machine learning researchers (or in fact, for anyone who does not spend years of her life studying quantum computational complexity). This is why I will only summarise the results from a high-level perspective here. 

\begin{itemize}
    \item Given access to a quantum algorithm that encodes data into quantum states, we can prepare a mixed quantum state $\rho$ representing a $M \times M$ kernel Gram matrix in time $\mathcal{O}(MN)$, where $N$ is the size of the inputs $\mathbf{x} \in \mathbb{R}^N$ (see \cite{rebentrost2014quantum} or \cite{schuld2018supervised} Section 6.2.5),
    \item We can prepare a quantum state $\ket{\mathbf{y}}$ representing $M$ binary labels as amplitudes in time $\mathcal{O}(M)$ (see for example \cite{iten2016quantum}, or \cite{schuld2018supervised} Section 5.2.1).
    \item Given $\ket{\mathbf{y}}$, as well as $k \in \mathcal{O}(\epsilon^{-1})$ ``copies'' of $\rho(\mathbf{x})$ (meaning that we have to repeat the first step $k$ times), we can prepare $\ket{\boldsymbol\alpha} = \rho^{-1}(\mathbf{x}) \ket{\mathbf{y}}$, a state whose amplitudes correspond to the coefficients $\boldsymbol{\alpha}$ in Theorem~(\ref{thm:convex}), to precision $\epsilon$ in time $\mathcal{O}(k \log d)$, where $d$ is the rank of $\rho$ (see \cite{lloyd2014quantum}, where this quantum algorithm was called ``quantum principal component analysis'', or \cite{schuld2018supervised} Section 5.4.3).
    \item We can estimate the amplitudes of $\ket{\boldsymbol\alpha}$ in time $\mathcal{O}(S/\tilde{\epsilon}^2)$ to precision $\tilde{\epsilon}$, where $S \leq M$ is the number of nonzero amplitudes (following from standard probability theory applied to quantum measurements, or \cite{schuld2018supervised} Section 5.1.3).
\end{itemize}
Overall, this is a recipe to compute the $S$ coefficients of the support vectors in time that is linear in the number of data points, a feat that is unlikely to be possible with a classical computer, at least not without imposing more structure on the problem, or allowing for heuristic results.

\end{document}